\documentclass[11pt]{article}

\usepackage[utf8]{inputenc}
\usepackage[T1]{fontenc}

\usepackage{fullpage}

\usepackage{url}            %
\usepackage{hyperref}       %
\usepackage{booktabs}       %

\usepackage{amsthm}  
\usepackage{amsmath}
\usepackage{amsfonts}
\usepackage{amssymb}
\usepackage{tikz} 
\usepackage{eqparbox}
\usepackage{mathtools}
\usepackage{cases}
\usepackage{xspace}
\usepackage{enumerate}
\usepackage{thmtools}
\usepackage{thm-restate}
\usepackage{multirow}
\usepackage{subcaption}
\usepackage{algorithm}
\usepackage{algorithmic}
\usepackage{cleveref}
\usepackage{color}
\usepackage{tcolorbox}
\usepackage{cite} %
\usepackage{xcolor}
\usepackage{xpatch}
\usepackage{bm}

\usetikzlibrary{calc,shapes.geometric, backgrounds}
\usetikzlibrary{decorations.pathreplacing,calligraphy}

\usepackage{libertine}
\usepackage[libertine]{newtxmath}
\usepackage{microtype}

\newcommand{\inputpreds}{input predictions\xspace}

\newcommand{\subr}{\textsc{Partial}\xspace}
\newcommand{\onst}{\mathrm{ON}^{\mathrm{ST}}}
\newcommand{\onsf}{\mathrm{ON}^{\mathrm{SF}}}
\newcommand{\onfl}{\mathrm{ON}^{\mathrm{FL}}}
\newcommand{\on}{\mathrm{ON}}
\newcommand{\hB}{{\bm\hat{B}}}

\newcommand{\algC}{C}
\newcommand{\optC}{C^*}
\newcommand{\optT}{T^*}
\newcommand{\predC}{{\bm\hat{C}}}
\newcommand{\appC}{{\bm\tilde{C}}}
\newcommand{\predT}{{\bm\hat{T}}}
\newcommand{\predR}{{\bm\hat{R}}}
\newcommand{\hx}{{\bm\hat{x}}}
\newcommand{\hr}{{\bm\hat{r}}}
\newcommand{\hs}{{\bm\hat{s}}}
\newcommand{\htt}{{\bm\hat{t}}}
\newcommand{\tlast}{r_{last}}
\newcommand{\tabort}{t_{abort}}
\newcommand{\Cabort}{C_{abort}}

\newcommand{\errhyp}{\ensuremath{\Lambda}\xspace}

\newcommand{\gammatsp}{\gamma^{\mathrm{TSP}}}
\newcommand{\gammadar}{\gamma^{\mathrm{DaRP}}}
\newcommand{\gammast}{\gamma^{\mathrm{ST}}}
\newcommand{\gammasf}{\gamma^{\mathrm{SF}}}
\newcommand{\gammafl}{\gamma^{\mathrm{FL}}}

\newcommand{\oldarp}{{\sc OlDARP}\xspace}
\newcommand{\oltsp}{{\sc OlTSP}\xspace}

\newcommand{\replan}{\textsc{Replan}\xspace}

\newcommand{\ignore}{\textsc{Ignore}\xspace}
\newcommand{\smart}{\textsc{SmartStart}\xspace}
\newcommand{\predreplan}{\textsc{PredictReplan}\xspace}
\newcommand{\mrin}{\textsc{Mrin}\xspace}

\newcommand{\OPT}{\ensuremath{\mathrm{OPT}}\xspace}
\newcommand{\ALGOHL}{\textsc{MrinTrust}\xspace} %
\newcommand{\ALGOGEN}{\textsc{DelayTrust}\xspace} %
\newcommand{\ALGOSMART}{\textsc{SmartTrust}\xspace}
\newcommand{\PR}{\textsc{PredReplan}\xspace} %
\newcommand{\cO}{\mathcal{O}}

\newcommand{\cA}{\mathcal{A}\xspace}

\newcommand{\cH}{\mathcal{H}}

      \theoremstyle{plain}
  \newtheorem{theorem}{Theorem}
  \newtheorem{lemma}[theorem]{Lemma}  
  \newtheorem{corollary}[theorem]{Corollary} 
  \newtheorem{proposition}[theorem]{Proposition}  
  
  \newtheorem{observation}{Observation}
  \theoremstyle{definition}

  {\bfseries}{\itshape}

\DeclarePairedDelimiter{\abs}{\lvert}{\rvert}

\usepackage[textsize=scriptsize,color=green,textwidth=2cm]{todonotes}

\definecolor{predred}{HTML}{F53349}
\definecolor{reqgreen}{HTML}{6BC199}

\definecolor{unexpblue}{HTML}{0E6792}
\definecolor{absentorange}{HTML}{FDBD09}

\tikzset{
    req/.style = {font=\footnotesize, fill=reqgreen, circle,inner sep=0pt,minimum width=0.45cm},
    pred/.style = {font=\footnotesize, draw=predred, circle,inner sep=0pt,minimum width=0.375cm,line width=0.75mm},
    corr/.style = {font=\footnotesize, draw=predred,circle,fill=reqgreen,inner sep=0pt,minimum width=0.45cm,line width=0.75mm},
    origin/.style = {black,circle,draw,minimum width=0.5cm},
    server/.style = {black,regular polygon,regular polygon sides=4,fill,inner sep=0pt,minimum width=0.5cm},
    absent/.style = {absentorange,line width=1.5pt},
    unexp/.style = {unexpblue,line width=1.5pt},
}

\newcommand{\convexpath}[2]{
  [   
  create hullcoords/.code={
    \global\edef\namelist{#1}
    \foreach [count=\counter] \nodename in \namelist {
      \global\edef\numberofnodes{\counter}
      \coordinate (hullcoord\counter) at (\nodename);
    }
    \coordinate (hullcoord0) at (hullcoord\numberofnodes);
    \pgfmathtruncatemacro\lastnumber{\numberofnodes+1}
    \coordinate (hullcoord\lastnumber) at (hullcoord1);
  },
  create hullcoords
  ]
  ($(hullcoord1)!#2!-90:(hullcoord0)$)
  \foreach [
  evaluate=\currentnode as \previousnode using \currentnode-1,
  evaluate=\currentnode as \nextnode using \currentnode+1
  ] \currentnode in {1,...,\numberofnodes} {
    let \p1 = ($(hullcoord\currentnode) - (hullcoord\previousnode)$),
    \n1 = {atan2(\y1,\x1) + 90},
    \p2 = ($(hullcoord\nextnode) - (hullcoord\currentnode)$),
    \n2 = {atan2(\y2,\x2) + 90},
    \n{delta} = {Mod(\n2-\n1,360) - 360}
    in 
    {arc [start angle=\n1, delta angle=\n{delta}, radius=#2]}
    -- ($(hullcoord\nextnode)!#2!-90:(hullcoord\currentnode)$) 
  }
}

\makeatletter
\xpatchcmd\algorithmic
  {\ALC@linenosize \arabic{ALC@line}\ALC@linenodelimiter}
  {\ALC@linenosize \theALC@line\ALC@linenodelimiter}
  {}{\fail}

\newcommand\setAlgoLinenoFormat[1]{%
  \renewcommand*{\theALC@line}{#1{ALC@line}}}
\makeatother

\title{A Universal Error Measure for Input Predictions Applied to Online~Graph Problems}

\author{Giulia Bernardini~\thanks{University of Trieste, Italy, and CWI, Amsterdam, The Netherlands. \emph{giulia.bernardini@units.it}. Partially supported by the Netherlands Organisation for Scientific Research (NWO) through project OCENW.GROOT.2019.015 ``Optimization for and with Machine Learning (OPTIMAL)''.}  
\and Alexander Lindermayr~\thanks{Faculty of Mathematics and Computer Science, University of Bremen, Germany. \emph{\{linderal,nicole.megow\}@uni-bremen.de}. Partially supported by the German Science Foundation (DFG) under contract 146371743 -- TRR 89 Invasive Computing.}
\and Alberto Marchetti-Spaccamela~\thanks{La Sapienza University of Rome, Italy, and INRIA-Erable, France. \emph{alberto@diag.uniroma1.it}.} 
\and Nicole Megow~\footnotemark[2]
\and
Leen Stougie~\thanks{CWI and Vrije Universiteit, Amsterdam, The Netherlands, and INRIA-Erable, France. \emph{Leen.Stougie@cwi.nl}. Supported by the Netherlands Organisation for Scientific Research (NWO) through Gravitation-grant NETWORKS-024.002.003 and through project OCENW.GROOT.2019.015 ``Optimization for and with Machine Learning (OPTIMAL)''.}
\and Michelle Sweering~\thanks{CWI, Amsterdam, The Netherlands. \emph{Michelle.Sweering@cwi.nl}. Supported by the Netherlands Organisation for Scientific Research (NWO) through Gravitation-grant NETWORKS-024.002.003.}
}

\date{}

\begin{document}

\maketitle

\begin{abstract}
    We introduce a novel measure for quantifying the error in input predictions. The error is based on a minimum-cost hyperedge cover in a suitably defined hypergraph and provides a general template which we apply to online graph problems. The measure captures errors due to absent predicted requests as well as unpredicted actual requests; hence, predicted and actual inputs can be of arbitrary size. We achieve refined performance guarantees for previously studied network design problems in the online-list model, such as Steiner tree and facility location. Further, we initiate the study of learning-augmented algorithms for online routing problems, such as the online traveling salesperson problem and the online dial-a-ride problem, where (transportation) requests arrive over time (online-time model). We provide a general algorithmic framework and we give error-dependent performance bounds that improve upon known worst-case barriers, when given accurate predictions, at the cost of slightly increased worst-case bounds when given predictions of arbitrary quality. 
\end{abstract}

\thispagestyle{empty}

\newpage

\section{Introduction}

We develop a novel measure for quantifying the error in input predictions %
and apply it to derive error-dependent performance guarantees of algorithms for online
metric graph problems.
Online graph problems are among the most fundamental online optimization problems, where an initially unknown input is revealed incrementally. The two main paradigms for the incremental information release are the {\em online-time} model and the {\em online-list} model. In the online-time model, initially unknown requests are revealed over time and can be served any time, whereas in the online-list model, requests are revealed one-by-one and must be served immediately before the next request appears. 
In this work, we address specifically the following online routing and network design problems.
\textbf{Online-time routing problems. \ }
In the classical \emph{Online Traveling Salesperson Problem} (\oltsp) and \emph{Online Dial-a-Ride Problem} (\oldarp), a server can move at unit speed in a given metric space. %
Transportation requests appear online over time, each defining a start and end point in the metric 
space~(in the TSP both points are equal). The task is to determine a tour %
to serve 
all requests~(in any order) %
by moving to the corresponding start and end points. The objective is 
to minimize the makespan, i.e., the time point when all requests have been served and the server is 
back in the origin. These problems are well-studied~\cite{AusielloFLST01,AscheuerKR00,FeuersteinS01,BjeldeHDHLMSSS21,BlomKPS01}, as well as other related variants~\cite{JailletW08,JailletW06,LipmannLPSS04}. %

\textbf{Online-list network design problems. \ }
In the \emph{Online Steiner Tree Problem}, requests are terminal nodes that are revealed one-by-one in a given metric space (typically represented as a complete edge-weighted graph) and must be connected to a fixed root by selecting edges via other (Steiner) nodes. In the closely related \emph{Online Steiner Forest Problem}, a request is composed of two nodes which have to be connected by the selected set of edges. In both problems, the objective is to minimize the total cost of selected edges. In the more general \emph{Online Facility Location Problem}, a facility can be opened at every vertex at a certain one-time cost at any time, and arriving client vertices are connected upon arrival to the closest open facility %
at the cost of the shortest path to it. %
The goal is to minimize the opening and connection cost. 
These problems are very well-studied~\cite{Meyerson01,Fotakis08,EisenstatMS14,ImaseW91,Umboh15,BermanC97,BamasDM22,AwerbuchAB96,WestbrookY93,Angelopoulos07,DehghaniEHLS18,AnagnostopoulosBUH04,AnNS17,BuchbinderCN14,FotakisKZ21}. 

The performance of online algorithms is typically assessed by %
worst-case analysis. %
An algorithm is called~$\rho$-{\em competitive} if it computes, for any input instance,  a solution with objective value within a multiplicative factor $\rho$ of the optimal value that can be computed when knowing the full instance upfront. 
The  {\em competitive ratio} of an algorithm is the smallest factor $\rho$ for which it is $\rho$-competitive. For the above problems (nearly) tight bounds on the competitive ratio are known. 
For \oltsp and \oldarp, there have been shown best possible $2$-competitive algorithms~\cite{AscheuerKR00,AusielloFLST01}. %
For the online network design problems, the existence of $O(1)$-competitive algorithms has been ruled out~\cite{ImaseW91,Fotakis08,Meyerson01} and algorithms with (tight) (poly-)logarithmic upper bounds have been shown~\cite{BermanC97,Umboh15,ImaseW91,Meyerson01,Fotakis08}.

The assumption in online optimization of not having any prior knowledge about future requests seems overly pessimistic. In particular, given the success of machine-learning methods and data-driven applications, one may expect to have access to predictions about future requests. However, simply trusting such predictions might lead to very poor solutions, as these predictions come with no quality guarantee. 
The recent vibrant line of research initiated in~\cite{LykourisV18,LykourisV21} %
aims at incorporating such error-prone predictions into online algorithms, to go beyond worst-case barriers. The goal are {\em learning-augmented algorithms} with a performance that is close to 
that of an optimal {offline} algorithm when given accurate predictions (called {\em consistency}) and, at the same time, 
never being (much) worse than that of a best known algorithm without access to predictions (called {\em robustness}). 
Further, the performance of an algorithm shall degrade in a controlled way with increasing prediction~error. 

In this paper, we consider an {\em \inputpreds} model, i.e., there is given a set $\predR$ of predictions for the actual online input $R$ of a problem. We do not make any assumption on the quality of the prediction or on its size. In particular, $\predR$ might be substantially larger or smaller than $R$.

Defining an appropriate error measure is a crucial task in this line of research. 
There is no common agreement (yet) in the literature on what constitutes a good error measure.
The philosophy behind our error measure is the following. Any learning-augmented algorithm needs to trust the predictions to some extent, as otherwise no improvement upon an online algorithm is possible.
The error should then be able to sensitively bound how much any (reasonable) algorithm pays for erroneous predictions.
Roughly, 
our error measure achieves this by approximating
the extra cost that an %
algorithm trusting the predictions has when it serves the true instance; 
very informally, this is $\OPT_{\text{trust} \predR}(R) - \OPT(R)$.
Several natural measures (for graph problems) have been proposed, such as the number of erroneous predictions~\cite{XuM22}, the $\ell_1$-norm (e.g., distances between predicted and real points), or more involved perfect matching-based errors~\cite{AzarPT22}.
Although we cannot expect that a single error measure is appropriate for all problems,
we propose a universal template %
based on the cost of a hyperedge cover in a bipartite hypergraph that is constructed in a problem-specific way. 

\subsection{Our contributions} 
\paragraph{Cover error for \inputpreds.}
Here we sketch the main idea of our error measure, which will be made precise in Section~\ref{sec:covererror}. We separately cover the errors incurred by \emph{unexpected actual requests},~$R \setminus \predR$, and 
\emph{absent predicted requests},~$\predR \setminus R$, as these pose a potential threat 
to an algorithm which trusts $\predR$. 
For each of the two error types we consider a suitable weighted bipartite hypergraph with 
erroneous requests on the left side and define an error measure combining the costs of minimum hyperedge covers of the left side of each hypergraph.
Let us concentrate on errors due to 
unexpected actual requests, being the nodes on the left side, with the predicted requests 
as the nodes on the right side. %
Each hyperedge links a single node on the right side with a subset of the nodes on the left side which it \emph{covers.} 
Its contribution to the overall error, %
i.e., its cost in the hyperedge cover problem to cover all left side nodes, is related to the optimal cost 
for the subinstance induced by its nodes. E.g. in \oltsp~this cost is the value of an 
optimal tour for some unexpected requests (left) when starting from some predicted request (right),
which can be seen as a minimum detour that needs to be made from the predicted request 
to serve the unexpected requests. 
Bounding the number of left side requests in the hyperedges by~$k$ %
yields a hierarchical family~$\{ \errhyp_k\} _{k=1}^\infty$ of errors, with higher values of~$k$ giving errors that reflect 
more precisely the cost due to trusting wrong predictions.

The cover error fulfills several useful properties. 
First, it provides a framework that may apply to various problems by assigning appropriate costs to the hyperedges.  
E.g. for the online-time model it allows to integrate in a very
precise way actual and predicted release dates, 
as we demonstrate in Section~\ref{sec:metric}. This is a feature which previous 
metric graph errors seem to miss~\cite{AzarPT22,XuM22}, because they rely on counting incorrect predictions or disallow asymmetric cost functions. Our error also naturally supports 
different sizes of $R$ and $\predR$, reflecting the input sequence length being unknown in almost all online optimization problems in the literature. 
Although previously studied error 
measures~\cite{AzarPT22,XuM22} do support this in theory, we will show in 
Section~\ref{sec:covererror} that they %
fail to detect good predictions 
in certain scenarios, which results in 
imprecise weak performance bounds. %
In contrast, the cover error guarantees an almost optimal performance of the same algorithms in these cases.
We therefore hope that the cover error will be useful for better analyzing existing 
learning-augmented algorithms and other problems in the future.

\paragraph{Algorithms with error-sensitive performance bounds.}
Our algorithmic results are twofold: we %
provide the first learning-augmented algorithms for online-time routing problems, %
and we give new error-dependencies for existing algorithms for online-list network design problems. 
The unifying element is that we achieve these by problem-dependent implementations %
of our new cover error.
We first introduce a general framework  for \oltsp and \oldarp, in which we delay the moment in which we start following the optimal predicted tour by a multiplicative trust factor $\alpha\in (0,1)$. 
For robustness, before starting to follow the predictions any $\rho$-competitive online algorithm is executed in a  black-box fashion. We prove an error-dependency w.r.t. the first cover error in the hierarchy $\Lambda_1$ (as properly defined in Section \ref{sec:covererror}). We denote by $\optC$ the cost of an optimal tour on the actual requests $R$.

\begin{theorem}\label{thm:algogen}
\oltsp and \oldarp admit learning-augmented algorithms that use a $\rho$-competitive algorithm as a subroutine and achieve a competitive ratio that is, for any~$\alpha > 0$%
, bounded by 
\[ 
   \min\left\{(1 + \alpha) \left(1 + \frac{3 \cdot \errhyp_1}{\optC}\right),  1 + \rho + \frac{\rho}{\alpha}  \right\}.
\]
\end{theorem}
Hence, sufficiently good predictions help to beat the classic lower bound of 2~\cite{AusielloFLST01} for \oltsp. %

When using the algorithm of~\cite{AscheuerKR00} as a subroutine, we can further refine our algorithm by carefully aligning the used waiting strategies and prove an improved robustness guarantee.

\begin{theorem}\label{thm:algosmart}
   \oltsp and \oldarp admit a learning-augmented algorithm that uses the 2-competitive algorithm of~\cite{AscheuerKR00} as a subroutine and achieves a competitive ratio that is, for any~$\alpha > 0$, bounded by 
    \[
        \min\left\{(1 + \alpha) \left(1 + \frac{3 \cdot \errhyp_1}{\optC}\right),  2 + \frac{2}{\alpha}  \right\}.
    \]
\end{theorem}

In general, %
online algorithms for \oltsp and \oldarp\ aim for tackling the 
uncertainty of the input rather than efficient running times, which is also the case for the 
above discussed results. 
Yet, we show in \Cref{sec:oltsp_extensions} that we can trade efficiency with slightly larger constant factors in the guarantees.

Further, we remark that simpler and tightened results are possible for restricted metric spaces: in \Cref{sec:oltsp_extensions} we provide improved bounds for \oltsp on the positive half of the real line. Here, a minimalistic prediction, a single value predicting the optimal makespan~$\optC$, suffices to obtain an almost tight consistency-robustness tradeoff.

We complement these theoretical bounds with empirical results (see \Cref{sec:experiments}) on simulated real-world taxi instances in the city road network of Manhattan, which indicate the superior performance of %
our %
new algorithms
compared to classic methods in both general and relevant restricted scenarios.

Further, we consider online-list graph problems and analyze the algorithmic framework provided by Azar, Panigrahi and Touitou~\cite{AzarPT22} w.r.t.~our new cover error. 
For each problem, we specify hyperedge cost functions %
which follow the same %
paradigm and prove new error-dependent bounds.

\begin{theorem}\label{thm:steiner-tree} 
The algorithms in~\cite{AzarPT22} for the online Steiner tree or online (capacitated) facility location problem %
incur%
, for any parameter $k \geq 1$, %
a cost of at most
\(
    \cO(1) \cdot \OPT + \cO(\log k) \cdot \errhyp_k.    
\)
\end{theorem}

\begin{theorem}\label{thm:steiner-forest}
    The algorithm in~\cite{AzarPT22} for the online Steiner forest problem incurs, for any parameter $k \geq 1$, %
    a cost of at most
    \(
        \cO(1) \cdot \OPT + \cO(k) \cdot \errhyp_k.    
    \)
\end{theorem}

These bounds hold \emph{simultaneously} for any $k$. %
The algorithm is still robust due to the robustness bound of $\cO(\log \abs{R})$ provided by Azar et al. in~\cite{AzarPT22}.
On the technical side, we can exploit a technical lemma by~\cite{AzarPT22} that allows us to split the analysis in two parts. The actual proofs for bounding the algorithm's cost by the cost of an optimal solution and the cover error are completely different.

For certain input scenarios we substantially strengthen the bounds on the competitive ratio provided by Azar et al.~\cite{AzarPT22}. %
Indeed, their bound never improves over $\Omega(\log(\max\{\abs{R}, \abs{\predR}\} - \min\{\abs{R}, \abs{\predR}\}))$, which is not better than the best possible competitive ratio $\cO(\log \abs{R})$ for classic online algorithms, if $\abs{\predR}$ and $\abs{R}$ differ significantly. %
	We show that there are input scenarios for which their error measure overestimates the actual error substantially and, thus, gives poor performance bounds while the algorithm performs actually well; more precisely, 
	we prove a constant competitive ratio for the algorithm in~\cite{AzarPT22}  w.r.t. our error measure whereas the bound w.r.t. the previous measure is $\cO(\log \abs{R})$. %

\subsection{Further related work}

While untrusted predictions have been successfully integrated into online models for many 
different problems, %
none of the previous approaches and models seem to 
capture the complexity of combined  routing and scheduling decisions. %
Related research includes 
work on scheduling %
\cite{LattanziLMV20,AzarLT21,AzarLT22,BamasMRS20,Mitzenmacher20,
ScullyGM22,PurohitSK18,Im0QP21,LindermayrM22}, routing %
in metric spaces such 
as the k-server problem and more generally metrical task 
systems~\cite{AntoniadisCE0S20,LindermayrMS22}, graph exploration~\cite{EberleLMNS22} 
and %
online network design~\cite{AzarPT22,XuM22,AlmanzaCLPR21}. %

Further, there is hardly any work on integrating untrusted predictions into online problems 
in the online-time model. The only exception seems to be the work by 
Antoniadis et al.~\cite{AntoniadisGS21speed-scaling} on online speed scaling with a prediction 
model that includes release dates and deadlines. The nature of the speed-scaling problem, 
however, is very different from the routing problems we consider. Other works on non-clairvoyant 
scheduling with jobs arriving over time (minimizing flow time~\cite{AzarLT21,AzarLT22}, 
total completion time~\cite{LindermayrM22} and energy~\cite{BamasMRS20}) 
assume predictions on the job sizes or priorities; release dates are known in advance 
in~\cite{BamasMRS20}, while \cite{LindermayrM22,AzarLT21,AzarLT22}
consider purely online problems w.r.t.~job~arrivals.

Very recently and independently of our work, %
two papers~\cite{HuWLCL22online,GouleakisLS22learning-augmented} were announced that also study \oltsp in the learning-augmented setting.
Hu et al.~\cite{HuWLCL22online} consider \oltsp with different prediction models in general metric spaces. For arbitrary input predictions, their result has no error-dependency and a weaker consistency-robustness tradeoff compared to~\Cref{thm:algosmart}. Gouleakis et al.~\cite{GouleakisLS22learning-augmented} exclusively study \oltsp on the real line. Assuming that the correct number of requests is known in advance, they study the power of  predictions on the locations; their results are incomparable to ours.

\subsection{Organization}

We first introduce the cover error in \Cref{sec:covererror}, and then give our results for online-time routing problems and online-list network design problems in \Cref{sec:metric} resp. \Cref{sec:online-list}. Finally, we present empirical experiments in~\Cref{sec:experiments}.

\section{The cover error}\label{sec:covererror}

Given an input prediction $\predR$ and the actual input $R$, we design an error measure that \emph{covers} every erroneously predicted item, i.e., all unexpected requests $R \setminus \predR$ and all absent predicted requests~$\predR \setminus R$. 
As a concrete example think of \oltsp where a learning-augmented algorithm trusts (at least to a certain degree) the predicted requests in $\predR$, and thus follows an optimal tour on $\predR$. 
After serving a predicted request $(\hx,\hr)$, it may serve some unexpected actual requests $R' \subseteq R \setminus \predR$ that have already been released and are \emph{relatively close} to $(\hx,\hr)$ (both time- and location-wise). In our terminology, think of~$(\hx,\hr)$ covering $R'$, and observe that the cost for this cover shall naturally be equal to the optimal cost for serving $R'$ when starting in $\hx$ at time $\hr$.
Conversely, the predicted requests that have not shown up and are nevertheless visited, $\predR \setminus R$, should be covered by actual requests, to make up for the extra cost incurred by these superfluous visits. 

We can arguably expect that any well-performing algorithm should be as least as good as serving all unexpected requests $R \setminus \predR$ and all absent predicted requests $\predR \setminus R$ %
in such partitions which can be covered by, respectively, predicted and actual requests in the cheapest possible way. 

We now embed this intuition into a precise definition.
Let $A$ and $B$ be two sets of (possibly different) size %
and let $k \geq 1$. 
We define a bipartite hypergraph~$G_k = (A \cup B, \cH)$ where $\cH$ is the set of all hyperedges which have exactly one endpoint in~$B$ and at most~$k$ endpoints in~$A$.
A $k$-\emph{hyperedge cover of $A$ by $B$} is a set of hyperedges $\cH' \subseteq \cH$ in $G_k$ such that every vertex in $A$ is incident to at least one hyperedge in~$\cH'$.
If every hyperedge $h \in \cH$ of $G_k$ has an associated cost $\gamma(h)$, a \emph{minimum-cost $k$-hyperedge cover~$\cH'$} is a $k$-hyperedge cover which minimizes the total hyperedge cost $\sum_{h' \in \cH'} \gamma(h')$. We denote the value of a min-cost $k$-hyperedge cover of $A$ by $B$ by $\Gamma_k(A, B)$.
Finally, the cover error, denoted by $\errhyp_k(R, \predR)$, is given by
\[
    \errhyp_k(R, \predR) = \Gamma_\infty(\predR, R) + \Gamma_k(R, \predR).    
\]

Notice that we allow arbitrary large hyperedges ($k = \infty$) to cover predicted requests $\predR$. We emphasize that all results also hold for a symmetric error definition $\Gamma_k(\predR, R) + \Gamma_k(R, \predR)$, because $\Gamma_i(A, B) \geq \Gamma_{i+1}(A, B)$, for any $i$. Nevertheless we use this asymmetric definition to obtain a stronger bound when covering $\predR$. Intuitively, this is possible because all predicted requests are known in advance (as opposed to the actual requests, which arrive online).

We simply write $\errhyp_k$ if $\predR$ and $R$ are clear from the context. Since our error measure shall give value zero %
if $\predR = R$, 
we require that the cost of every hyperedge $\{a\} \cup \{b\}$ for some $a \in A$ and~$b \in B$ is equal to zero if~$a = b$. Then, all vertices in~$A \cap B$ can be covered trivially by $B$,
and we conclude that $\Gamma_k(A \setminus B, B) = \Gamma_k(A, B)$. \Cref{fig:hyperedge-cover} depicts an example of a $k$-hyperedge cover.

\begin{figure}
    \begin{subfigure}[b]{0.35\textwidth}
        \begin{center}           
        \begin{tikzpicture}[scale=0.75]
            \node[corr] (p1) at (0.5,0.5) {1};
            \node[pred]      at (2.5,1) {2};
            \node[pred]      at (3.5,0.5) {3};
            \node[pred] (p2) at (4,4) {5};
            \node[pred]      at (0,3) {4};
            \node[pred]      at (0,4) {6};
    
            \node[req] (r1) at (0.5,1.5) {7};
            \node[req] (r2) at (3,4.5) {8};
            \node[req] (r3) at (2,4) {9};
            \node[req] (r4) at (3,3) {10};
            \node[req] (r5) at (5,3) {11};
            \node[req] (r6) at (4,2.5) {12};

        \begin{pgfonlayer}{background}
            \draw[unexp, fill=unexpblue!10!,fill opacity=0.5] \convexpath{p2,r5,r6,r4}{12pt};
            \draw[unexp, fill=unexpblue!10!,fill opacity=0.5] \convexpath{p2,r3,r2}{12pt};
            \draw[unexp, fill=unexpblue!10!,fill opacity=0.5] \convexpath{p1,r1}{12pt};
        \end{pgfonlayer}
        \end{tikzpicture}
        \end{center}
        \caption{Instance in metric space}
    \end{subfigure}
    \hfill
    \begin{subfigure}[b]{0.3\textwidth}
        \begin{center}     
        \begin{tikzpicture}[scale=0.75]
            \node[pred]  (p1)    at (4,1 * 0.75) {1};
            \node[pred]      at (4,2 * 0.75) {2};
            \node[pred]  at (4,3 * 0.75) {3};
            \node[pred] (p2) at (4,5 * 0.75) {5};
            \node[pred]      at (4,4 * 0.75) {4};
            \node[pred]      at (4,6 * 0.75) {6};
    
            \node[req] (r1) at (2,1 * 0.75) {7};
            \node[req] (r2) at (2,2 * 0.75) {8};
            \node[req] (r3) at (2,3 * 0.75) {9};
            \node[req] (r4) at (2,4 * 0.75) {10};
            \node[req] (r5) at (2,5 * 0.75) {11};
            \node[req] (r6) at (2,6 * 0.75) {12};
    
            \node at (2, 0) {$R \setminus \predR$};
            \node at (4, 0) {$\predR$};
    
        \begin{pgfonlayer}{background}
            \draw[unexp, fill=unexpblue!10!,fill opacity=0.5] \convexpath{p2,r4,r6}{11pt};
            \draw[unexp, fill=unexpblue!10!,fill opacity=0.5] \convexpath{p2,r2,r3}{11pt};
            \draw[unexp, fill=unexpblue!10!,fill opacity=0.5] \convexpath{p1,r1}{11pt};
        \end{pgfonlayer}
    
        \draw [line width=1pt,black,decorate, decoration = {calligraphic brace,raise=15pt,amplitude=4pt}] (r4.south) -- (r6.north) node[pos=0.5, left=0.7cm,black]{$\leq k$};
    
        \end{tikzpicture}
        \end{center}
        \caption{Min-cost $k$-hyperedge cover}
    \end{subfigure}
    \hfill
    \begin{subfigure}[b]{0.3\textwidth}
        \begin{center}     
        \begin{tikzpicture}[scale=0.75]
            \node[pred] (p1) at (4,1 * 0.75) {1};
            \node[pred]      at (4,2 * 0.75) {2};
            \node[pred]  at (4,3 * 0.75) {3};
            \node[pred] (p2) at (4,5 * 0.75) {5};
            \node[pred]      at (4,4 * 0.75) {4};
            \node[pred]      at (4,6 * 0.75) {6};
    
            \node[req] (r0) at (2,1 * 0.75) {1};
            \node[req] (r1) at (2,2 * 0.75) {7};
            \node[req] (r2) at (2,3 * 0.75) {8};
            \node[req] (r3) at (2,4 * 0.75) {9};
            \node[req] (r4) at (2,5 * 0.75) {10};
            \node[req] (r5) at (2,6 * 0.75) {11};
            \node[req] (r6) at (2,7 * 0.75) {12};
    
            \node at (2, 0) {$R$};
            \node at (4, 0) {$\predR$};
    
        \begin{pgfonlayer}{background}
            \draw[unexp, fill=unexpblue!10!,fill opacity=0.5] \convexpath{p2,r4,r6}{11pt};
            \draw[unexp, fill=unexpblue!10!,fill opacity=0.5] \convexpath{p2,r2,r3}{11pt};
            \draw[unexp, fill=unexpblue!10!,fill opacity=0.5] \convexpath{p1,r1}{11pt};
        \end{pgfonlayer}
    
        \draw [line width=1pt,black,decorate, decoration = {calligraphic brace,raise=15pt,amplitude=4pt}] (r4.south) -- (r6.north) node[pos=0.5, left=0.7cm,black]{$\leq k$};
    
        \draw[unexp, dotted] (r0) -- (p1);
        \end{tikzpicture}
        \end{center}

        \caption{$\Gamma_k(R, \predR) = \Gamma_k(R \setminus \predR, \predR)$}
    \end{subfigure}
    \caption{Example for a metric instance and input prediction with a min-cost $k$-hyperedge cover of the set of unexpected requests $R \setminus \predR$. The actual requests are filled green and the predicted requests are encircled red. The labels show which points in the metric space correspond to which nodes in the bipartite graphs.}
    \label{fig:hyperedge-cover}
\end{figure}
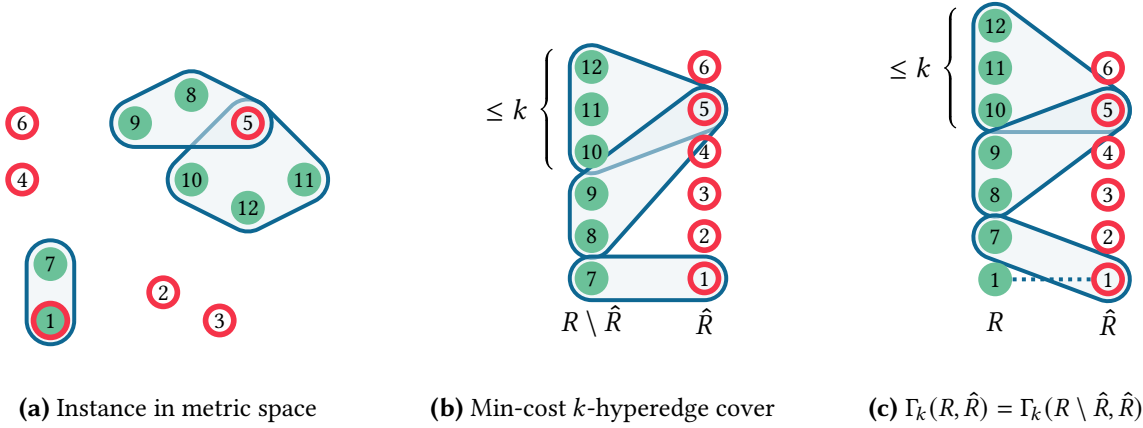

It remains to specify the cost $\gamma(A', b)$ of a hyperedge $A' \cup \{b\}$. 
Although we will give precise definitions separately for every concrete problem, all definitions follow a certain paradigm.
That is, the cost $\gamma(A', b)$ shall be equal to \emph{the value of an optimal solution for the subinstance induced by~$A'$ with respect to $b$}. This anchoring requirement is the \emph{single} detail which has to be specified for a concrete problem.
Note that this matches our intuition discussed above for \oltsp.

\paragraph{Comparison to other error measures.}
We compare the cover error to previously proposed error measures %
for the (undirected) online Steiner tree problem.
Xu and Moseley~\cite{XuM22} %
define a prediction error $\eta = \max\{\abs{\predR}, \abs{R}\} - \abs{\predR \cap R}$, the number of erroneous requests,  %
and prove that their algorithm is~$\cO(\log(\min\{\abs{R}, \eta \}))$-competitive. 
Azar et al.~\cite{AzarPT22} introduce the 
\emph{metric error with outliers}~$\lambda=(\Delta, D)$, where~$D$ is the value of a 
min-cost perfect matching between two equally sized subsets of~$R$ and~$\predR$, 
and~$\Delta$ is the total number of unmatched points in~$R$ and~$\predR$. They prove for their algorithms a multiplicative 
error dependency w.r.t.~$\log( \min\{\abs{R}, \Delta \} )$ and an additive error dependency w.r.t.~$D$.

We give a family of instances %
with~$n$ actual requests and an input prediction for which the algorithms of Azar et al.~\cite{AzarPT22} and Xu and Moseley~\cite{XuM22} perform arguably well, but their error measures and analyses yield a bound of $\cO(\log n) \cdot \OPT + \cO(\epsilon)$, which could be achieved even without predictions. For some $\epsilon > 0$, the instance is composed of one terminal request at $x_1$ and $n-1$ requests in an $\epsilon$-ball around the Steiner point $x_2$, but no request is exactly on~$x_2$. Both $x_1$ and $x_2$ are predicted. We can immediately observe that the number of erroneous requests~\cite{XuM22} is $\eta = n - 1$. For the metric error with outliers~\cite{AzarPT22}, note that any perfect matching is composed of at most two matches, therefore~$\Delta \geq n - 2$ %
and $D = \cO(\epsilon)$.
On the other hand, our cover error is bounded by $\errhyp_k \leq \errhyp_1 = \cO(n \cdot \epsilon)$ for any~$k$, because %
$x_2$ covers all requests in the $\epsilon$-ball around it. %
Then, \Cref{thm:steiner-tree} concludes that the algorithm of Azar et al.~\cite{AzarPT22} is indeed constant competitive for this instance when $\epsilon \to 0$.

\section{Online metric TSP with predictions}\label{sec:metric}
Let $M=(X,d)$ be a metric space, consisting of a set of points~$X$, with %
origin $o\in X$ and a metric~$d$. %
In the \emph{Online Metric Traveling Salesperson Problem} (\oltsp), a set of unknown %
requests $R$ is released online over time. 
A request~$(x,r)$ is composed of a point $x \in X$ and a release date $r\in \mathbb{R}_{\geq 0}$, i.e., the time at which the request becomes known and %
can be served. %
The task of  an algorithm %
is to route a server, which is initially in the origin and moves at unit speed, through all requests back to the origin.  %
The objective is to minimize the makespan, i.e., the total time required for this task.

The \oltsp \emph{with predictions}, is an \oltsp in which we are given additionally an a priori prediction~$\predR$ on the set of requests. %
We assume that the server receives a signal %
when it is back at the origin after serving all the actual requests. Unlike in the classic \oltsp, this is important, %
as otherwise an algorithm might continue considering predicted requests, thus, ruling out any robustness.

To specify our cover error, we define the cost of a hyperedge $R' \cup \{(x',r')\}$ as the extra cost of serving %
erroneous (unexpected or predicted absent) requests $R'$ w.r.t.~a request $x'$ (actual or predicted):
\begin{list}{}{}
    \item[$\gammatsp(R',(x',r'))=$]  \emph{optimal makespan for serving instance $R'$ from origin $x'$ and initial time $r'$.} %
\end{list}

To get some intuition, consider $X = \mathbb{R}_{\geq 0}$ and an algorithm that does not move before time $t$ if there are no requests. It will receive an adversarial request $(t,t)$ and hence encounters a ratio of $\frac{3}{2}$ \cite{BlomKPS01}. To overcome this when having (almost) accurate predictions, the server has to move towards (predicted) requests before they actually arrive. However, this pre-moving technique brings several challenges.
If an algorithm moves the server without interruption to a predicted request at the beginning of the instance, an adversary would immediately spawn the single actual request at the origin, giving an unbounded robustness. If the server would directly move back, one can similarly argue that the consistency is at least $\frac{3}{2}$.
Therefore, the key is to define a proper waiting strategy before moving towards predicted requests. 
We show that we can execute an arbitrary online algorithm while  %
delaying pre-moving to gain information about the instance.
This is very delicate, since too much delay clearly weakens the consistency, but too little delay gives weak robustness. 
In~\Cref{app:lower_bound} we prove the following result.

\begin{restatable}{theorem}{ThmTradeoffLB}
    \label{thm:tradeoff-lb}
    Let~$\alpha \in (0,1/2)$ and let $\cA$ be a $(1 + \alpha)$-consistent deterministic learning-augmented algorithm for \oltsp. Then, $\cA$ can be $\beta$-robust only for $\beta\geq \frac{1}{\alpha} - 1$. This holds even on the half-line. 
\end{restatable}

Our final algorithm uses a hyperparameter $\alpha > 0$ to configure the waiting duration and thereby achieves a tight asymptotic consistency-robustness tradeoff.
Intuitively, we can express our confidence in the prediction using~$\alpha$ and get customized guarantees.

\subsection{A general framework for \oltsp with predictions}\label{sec:blackbox}

Our strategy involves an initial %
delay phase in which we follow %
an arbitrary online algorithm, up to some predetermined time %
depending on the cost $\predC$ of an optimal tour $\predT$ of the predicted requests $\predR$.
After that, 
we start following $\predT$, %
adjusting it whenever the actual requests deviate from the predictions. We call this greedy strategy \predreplan (\PR for short), due to the analogy with the classic \replan heuristic~\cite{AscheuerKR00}. Let $p(t)$ be the server's location at time $t$.

\begin{algorithm}[H]
\caption{\PR}
\begin{algorithmic}
\STATE Follow $\predT$. Whenever an unexpected request $(x,r)$ is released, recompute and follow a fastest tour from $p(r)$ %
to the origin serving all unserved predicted requests as well as all the unserved unexpected requests. If the server receives an end signal in the origin, terminate.
\end{algorithmic}
\end{algorithm}

While this algorithm might move towards predicted requests which are known to be absent to make the analysis clearer, a practical implementation ignores these and thereby only improve its performance.
We formally define the class of algorithms \ALGOGEN, %
that is parameterized by our trust parameter $\alpha >0$, which scales the delay. Let $\cA$ be any $\rho$-competitive online algorithm for \oltsp. %
\begin{algorithm}[H]
    \caption{\ALGOGEN}
    \begin{algorithmic}[1]
        \setAlgoLinenoFormat{\roman}
        \STATE Follow $\cA$ as long as for time~$t$ it holds $t \leq \alpha \predC - d(p(t), o)$ 
        \STATE Move the server to the origin %
        \STATE Follow the \PR strategy until the end
    \end{algorithmic}
\end{algorithm}

We refer to the execution of each line as a \emph{phase}. %
We now prove the main \Cref{thm:algogen} for \oltsp, by showing for \ALGOGEN separately an error-dependent bound in Lemma~\ref{lem:competitive_general1} and a robustness  bound in Lemma~\ref{lem:competitive_general2}. 
Given an \oltsp instance, we denote by $\optC$ the makespan of an optimal tour~$\optT$ serving all actual requests. 

\begin{lemma}\label{lem:competitive_general1}
    \ALGOGEN has a competitive ratio of at most $(1 + \alpha) \left(1 + 3 \cdot \frac{ \errhyp_1}{\optC}\right)$, for any~$\alpha \geq 0$. %
\end{lemma}

\begin{proof}
We first bound~$\predC$. Fix a min-cost $\infty$-hyperedge cover of~$\predR$ by $R$ and an optimal tour~$\optT$ for~$R$. %
For every hyperedge $\predR'\cup \{(x,r)\}$ in the cover, we %
extend~$\optT$ by adding the optimal offline \oltsp tour for~$\predR'$ which starts at~$x$ at the time~$t$ at which~$\optT$ serves~$x$. %
Note that, since~$r \leq t$, the makespan of this subtour is bounded by the cost of $\predR'\cup \{(x,r)\}$. Since every predicted request is covered by at least one hyperedge, the constructed tour serves~$\predR$ and we conclude that
\(
    \predC \leq \optC + \Gamma_\infty(\predR, R).
\)

We now bound the makespan of the tour of the algorithm.
If the algorithm terminates in Phases~(i) or~(ii), its makespan is at most
\(
    \alpha \predC \leq \alpha \cdot (\optC + \Gamma_\infty(\predR, R)) \leq (1 + \alpha) \cdot (\optC + \errhyp_1).
\)

Otherwise, the algorithm reaches Phase~(iii). There it first 
computes an optimal tour~$\predT$ of length at most~$\predC$ serving all unserved predicted requests. The makespan only increases when unexpected requests arrive.
To this end, fix a min-cost $1$-hyperedge cover of~$R$ by~$\predR$ 
and a hyperedge~$\{(x',r')\} \cup \{(\hx,\hr)\}$ of this cover. 
We upper bound the additional cost due to $(x',r')$ by the cost of an excursion from the algorithm's current tour serving~$(x',r')$ . The algorithm might find a faster tour to serve all unserved requests and henceforth uses that. We distinguish two cases depending on the algorithm's 
remaining tour before request~$(x',r')$ arrived. 
If %
$\hx$ is not part of this tour, we consider an excursion which immediately deviates from $p(r')$ to serve $(x',r')$ and then returns to $p(r')$. By the triangle inequality, the length of this excursion is bounded
by twice the distance between $p(r')$ and $\hx$, plus the cost for optimally serving~$(x',r')$ from $\hx$ when starting at time $\hr$. Due to our assumption, $(\hx,\hr)$ must have already been served at some time $t$ with $r' \geq t \geq \hr$. Thus, the algorithm's server is at most~$r' - \hr$ units away from~$\hx$ at time $r'$, and the total time incurred for this excursion is bounded by
\[
    2 \cdot (r' - \hr) + \gammatsp(\{(x',r')\}, (\hx,\hr)) \leq 3 \cdot \gammatsp(\{(x',r')\}, (\hx,\hr)).
\]
Note that the inequality is due to the fact that $(x',r')$ can only be served after its %
release date. 

In the other case, the algorithm's server will visit $\hx$ at some later point in time, especially at least once after time $\hr$. We thus wait %
until the algorithm reaches
$\hx$ at some time $t \geq \hr$, 
and then serve $(x',r')$ using at 
most~$\gammatsp(\{(x',r')\}, (\hx,\hr))$ additional time. See \Cref{fig:oltsp-serve-unexpected-requests} for an illustration of both cases.

\begin{figure}
    \begin{subfigure}[b]{0.61\textwidth}
        \begin{center}           
        \begin{tikzpicture}[scale=0.85]
            \node[origin] (o) at (3,3) {o};

            \node[pred] (1) at (-2,0.5) {1}; 
            \node[pred] (2) at (0,0.5) {2}; 
            \node[pred] (3) at (2,2) {3}; 
            \node[pred] (4) at (-2,3) {4}; 
            \node[pred] (5) at (1,3) {5}; 

            \node[server] (s) at (1,1.25) {};

            \node[req] (6) at (-1.25,1.25) {6};
            \node[req] (7) at (3,1) {7};

        \draw[predred, line width = 1.5pt, solid] (o) -- (3) -- (s);
        \draw[predred, line width = 1.5pt, dashed] (s) -- (2) -- (1) -- (4) -- (5) -- (o);
        \draw[unexp] 
            (3) -- (7) %
            (1) -- (6);

        \draw [line width=1pt,unexpblue, decorate, decoration = {calligraphic brace,raise=3pt}] (s.150) -- (3.160) node[pos=0.5, above left=4pt,black]{$\leq r' - \hr_3$};

        \node at (4.5,1.6) {$\gamma \left(\{7\}, (\hx_3,\hr_3) \right)$};

        \end{tikzpicture}
        \end{center}

        \caption{The server (square) plans (dashed) and follows %
        (solid) the predicted tour. Two hyperedges are completely released: \{7\} is covered by 3, which is already served, and $\{6\}$ is covered by 1, which will be served in the future. Note that neither 1, 8 nor 9 have been released yet. Also notice that \PR might find a faster tour.}
    \end{subfigure}
    \hfill
    \begin{subfigure}[b]{0.43\textwidth}
        \begin{center}     
        \begin{tikzpicture}[scale=0.85]
            \node[pred] (p1) at (2,1 * 0.75) {1}; tour through predicted request
            \node[pred] (2) at (2,2 * 0.75) {2}; 
            \node[pred] (3) at (2,3 * 0.75) {3}; 
            \node[pred] (4) at (2,4 * 0.75) {4}; 
            \node[pred] (5) at (2,5 * 0.75) {5}; 

            \node[req] (r1) at (0,1 * 0.75) {1};
            \node[req] (6) at (0,2 * 0.75) {6};
            \node[req] (7) at (0,3 * 0.75) {7};
            \node[req] (8) at (0,4 * 0.75) {8};
            \node[req] (9) at (0,5 * 0.75) {9};
    
            \node at (0, 0) {$R$};
            \node at (2, 0) {$\predR$};
    
        \begin{pgfonlayer}{background}
            \draw[unexp, fill=unexpblue!10!,fill opacity=0.5] \convexpath{3,7}{9.5pt};
            \draw[unexp, fill=unexpblue!10!,fill opacity=0.5] \convexpath{5,8}{9.5pt};
            \draw[unexp, fill=unexpblue!10!,fill opacity=0.5] \convexpath{5,9}{9.5pt};
            \draw[unexp, fill=unexpblue!10!,fill opacity=0.5] \convexpath{p1,6}{9.5pt};
        \end{pgfonlayer}
          
        \draw[unexp, dotted, line width = 1.5pt] (r1) -- (p1);        
        \end{tikzpicture}
        \end{center}
        \caption{Min-cost $1$-hyperedge cover of $R$}
    \end{subfigure}
    \caption{}
    \label{fig:oltsp-serve-unexpected-requests}
\end{figure}

Since every actual request is covered by one hyperedge, we conclude that 
Phase~(iii) takes time at most
\(
    \predC + 3 \cdot \Gamma_1(R, \predR).
\)
Adding the time for Phases~(i) and~(ii) gives a makespan of at most
\[
    (1 + \alpha) \predC + 3 \cdot \Gamma_1(R, \predR) 
    \leq (1 + \alpha) \left( \optC + \Gamma_\infty(\predR, R) + 3 \cdot \Gamma_1(R, \predR) \right)
    \leq (1 + \alpha) \left( \optC + 3 \cdot \errhyp_1 \right). \qedhere 
\]
\end{proof}

\begin{lemma}\label{lem:competitive_general2}
\ALGOGEN has a competitive ratio of at most $1 + \rho + \frac{\rho}{\alpha}$, for any~$\alpha > 0$ and any~$\rho$-competitive algorithm used in Phase~(i).
\end{lemma}
\begin{proof}
If the algorithm terminates during Phase~(i) or~(ii), the competitive ratio is $\rho$. We are guaranteed to finish in one of these two phases if $\rho \optC \leq \alpha\predC$.

If we terminate within Phase~(iii), then $\predC < \frac{\rho}{\alpha} \optC$. Once the last request has arrived at some time~$\tlast \leq \optC$, our tour stays fixed.  We distinguish two cases. %
If the last request arrives before the end of Phase~(ii), then the cost of our tour comprises of the cost for finishing Phase~(ii), which is at most  $\alpha \predC$, and the cost of \PR for following the predicted tour, including all unexpected yet unserved requests, which is at most $\predC + \optC$. The total cost is thus bounded from above by 
\[ 
    \alpha \predC + \predC + \optC \leq \left( 1+ (1+\alpha)\cdot \frac{\rho}{\alpha} \right) \cdot  \optC = 
\left( 1+\rho + \frac{\rho}{\alpha} \right) \cdot \optC. 
\]

In the second case, the last request arrives in Phase~(iii). In this case the 
cost after $\tlast$ is the cost of following the predicted tour, adapted for incorporating the unexpected, yet unserved, requests. This is bounded above by the cost of returning to the origin, following the predicted tour $\predT$, and finally following the optimal tour, $T^*$. Note that the cost of returning to the origin is at most~$\tlast-\alpha \optC$. 
Hence, we %
complete the proof by upper bounding the makespan, for any $\rho\geq 1$, by 
\[
\tlast + \tlast- \alpha \predC + \predC + \optC \leq \left( 3 + (1-\alpha)\cdot \frac{\rho}{\alpha} \right) \optC = \left( 3-\rho + \frac{\rho}{\alpha}\right) \optC \leq \left( 1+\rho + \frac{\rho}{\alpha} \right) \optC. \qedhere
\]
\end{proof}

\subsection{Extensions and improvements}
\label{sec:oltsp_extensions}
\paragraph{An improved algorithm for \oltsp with predictions.} 
\label{sec:improved_robustness}
A best possible online algorithm for \oltsp is \smart, which is 2-competitive~\cite{AscheuerKR00}. Using this in Phase~(i) of \ALGOGEN, \Cref{thm:algogen} yields a robustness factor of at most $3 + \frac{2}{\alpha}$. We exploit \smart's waiting strategy to serve yet unserved requests and expedite Phase~(iii) avoiding unnecessary waiting time, obtaining an algorithm, \ALGOSMART, with improved robustness factor $2 + \frac{2}{\alpha}$. See \Cref{thm:algosmart} for \oltsp in~\Cref{app:smarttrust}.

\paragraph{Algorithms with polynomial running time.}
Algorithms \ALGOGEN and \ALGOSMART\  
require the computation of optimal TSP tours on subinstances. NP-hardness of TSP prohibits polynomial running time, unless P$=$NP. We provide performance guarantees for our learning-augmented algorithm framework when using polynomial-time $\nu$-approximation algorithms for solving TSP, which guarantee to find a TSP tour within a factor $\nu$ of the optimum.
We use a modified efficient \PR strategy which uses a $\nu$-approximate solution instead of an optimal solution and further ensures that errors due to such approximations do not add up too much compared to our error budget $\errhyp_1$. Adjusting the proof of \Cref{thm:algogen} yields the following result, whose proof is in~\Cref{app:polyalgs}. %
\begin{restatable}{theorem}{thmPolyRunningTime}\label{thm:competitive_general_polyn}
Given a $\nu$-approximation algorithm for metric TSP, the competitive ratio of the polynomial time \ALGOGEN using a polynomial time $\rho$-competitive online algorithm in Phase~(i) is, for any~$\alpha > 0$%
, bounded by 
\[
    \min\left\{(1 + \nu) (1 + \alpha) \left(1 + \frac{3 \cdot \errhyp_1}{2 \cdot \optC}\right),  \rho + (1 + \nu) \left(1+ \frac{\rho}{\alpha} \right)  \right\}.
\]
\end{restatable}

\paragraph{Online metric Dial-a-Ride with predictions.}
The \emph{Online Metric Dial-a-Ride Problem} (\oldarp) is a generalization of \oltsp where each request $(x^s,x^d,r)$ %
has a starting location $x^s$ and a destination~$x^d$. %
To serve a request, %
the server must first visit $x^s$ at some time not earlier than $r$, and then $x^d$. We assume that the server can carry at most one request at the time and 
cannot store it after pickup.

We show in \Cref{app:oldarp} that slight modifications of \ALGOGEN and \ALGOSMART yield \Cref{thm:algogen,thm:algosmart} for this  generalized setting. We define the cost function $\gammadar$ for the cover error: %
\begin{list}{}{}
    \item[$\gammadar(R',(x^s,x^d,r))=$] $\min \{\gammatsp(R', (x^s,r)), \gammatsp(R', (x^d, r + d(x^s, x^d)))\} + D$,
\end{list}
where $D$ is the maximum transportation distance in $R \cap \predR$.
Intuitively, an excursion %
can start %
from~$x^s$ after time $r$ or from $x^d$ after time $r + d(x^s,x^d)$ to serve %
$R'$, 
whatever is the shorter of the two. 
\paragraph{An improved algorithm for \oltsp on the half-line metric.}
When restricting the metric space to~$X = \mathbb{R}_{\geq 0}$, the best possible online algorithm %
is $\frac{3}{2}$-competitive~\cite{BlomKPS01}.
We design a learning-augmented algorithm tailored to this metric space and a minimalistic prediction, namely, a single value 
$\predC$ predicting the optimal makespan~$\optC$. We prove (\Cref{app:half_line}) the following error-dependent performance bound, which gives an almost tight consistency-robustness tradeoff w.r.t \Cref{thm:tradeoff-lb}.
\begin{restatable}{theorem}{thmHalfLine}
    There is a learning-augmented algorithm for the half-line metric that has for every $\alpha \in (0,1/2]$ a competitive ratio of at most
    \[
       \min\left\{ (1 + \alpha) \left( 1 + \frac{\errhyp_1}{\optC} \right), \frac{3}{2 \alpha}  \right\}.
    \]
\end{restatable}

\section{Online network design problems with predictions}\label{sec:online-list}

This section %
sketches the applicability of our new error measure %
for the online-list problems Steiner Tree, Steiner Forest and (capacitated) facility location. %
To prove new error-dependent performance bounds as stated in \Cref{thm:steiner-tree} and \Cref{thm:steiner-forest}, we revisit the algorithms proposed by Azar et al.~\cite{AzarPT22} and analyze it w.r.t.~our cover error measure. 
The key is an appropriate, problem-specific definition of the cost of a hyperedge and the corresponding analysis.
Recall that the cost of a hyperedge $R' \cup \{x'\}%
$ should (intuitively) be equal to the value of an optimal solution %
for serving a set of unexpected (predicted absent) requests %
$R'$ %
w.r.t.~a predicted (actual) request %
$x'$. %
We define the cost functions:

\begin{list}{}{}
    \item[\textbf{Steiner tree:}] $\gammast(R',x')=\ $\emph{cost of an optimal Steiner tree for terminals $R'$ with root $x'$.} %

    \item[\textbf{Steiner forest:}] $\gammasf(R',x')=\ $\emph{cost of an optimal Steiner forest
    for terminal pairs $R'$ when connecting via %
    $x'=(s',t')$ is free.}

    \item[\textbf{Facility location:}] $\gammafl(R',x')=\ $\emph{
    cost for opening facility $x'$ and assigning clients $R'$ to it.
    }
\end{list}

Our main technical contribution for online network design problems are the proofs of \Cref{thm:steiner-tree} and \Cref{thm:steiner-forest}.
We now give some intuition on how these proofs work by considering the online Steiner tree problem, and defer details and results for this and the other problems to \Cref{app:networkdesign}.

On a high level, the algorithm by Azar et al.~\cite{AzarPT22} for the online Steiner tree problem does the following. Each new request (terminal) is connected to the current solution greedily by buying edges on a shortest path to a vertex of the current tree. When the Greedy cost increased sufficiently (with thresholds following a doubling-strategy), the algorithm spends a certain budget (depending on the spent Greedy cost) on connecting as many future predicted requests as possible to the current solution.

The proof of \Cref{thm:steiner-tree} splits the execution of this algorithm into two parts, where the first part considers the time until all predicted requests are satisfied, and the second part the remaining execution. %
We then use a sub-result provided by Azar et al.~\cite{AzarPT22} which roughly bounds the total cost of the first part by the optimal solutions of $R$ and $\predR$, and the total cost of the second part by the cost of the algorithm for serving specific subsequences of the request sequence.
To further bound the first part, we use the structure of a min-cost $\infty$-hyperedge cover of $\predR$ to prove an upper bound of at most $\cO(1) \cdot \OPT + \cO(1) \cdot \Gamma_\infty(\predR, R)$. For the second part, we consider the total cost the Greedy %
algorithm incurs for %
a hyperedge of a min-cost $k$-hyperedge cover of $R$,
and conclude by the bounded hyperedge size and Greedy properties %
that this %
is at most $\cO(\log k)$ times the hyperedge cost, yielding a total bound of $\cO(1) \cdot \OPT + \cO(\log k) \cdot \Gamma_k(R, \predR)$. Here we especially use the fact that in our chosen partition of the algorithm's execution, any predicted terminal $\hx$ which covers actual requests must have already been served in the first part.

\section{Experiments}\label{sec:experiments}

We performed various empirical experiments
on real-world \oltsp instances that demonstrate the benefits of using our algorithms over classic online algorithms. %
We consider the road network of Manhattan~\cite{osm,boeing17osmnx} and compose $100$ instances of $10$ requests each based on taxi pickup requests from a dataset offered by the NYC Taxi \& Limousine Commission.\footnote{\url{https://www1.nyc.gov/site/tlc/about/tlc-trip-record-data.page}, downloaded 02/05/22} We compare \ALGOSMART with the classic online algorithms \replan~\cite{AscheuerKR00}, \ignore~\cite{AscheuerKR00,FeuersteinS01,ShmoysWW95} and \smart~\cite{AscheuerKR00}; all algorithms use efficient TSP heuristics.
We report for every experiment and instance the empirical competitive ratio, i.e. the average ratio between the algorithms performance and the approximated value of the optimal makespan, as well as error bars that denote the~95\% confidence interval over all instances.

We sketch here two relevant experiments and defer further details %
to \Cref{app:experiments}. The first experiment considers synthetic predictions with Gaussian noise $\sigma$ only in the request locations, i.e., the release dates are predicted correctly. %
The results (\Cref{fig:noise-in-locations-main}) show that \ALGOSMART with $\alpha = 0.1$ dominates classic algorithms even for arbitrarily bad predictions. In the second experiment only a part of the actual instance is predicted, which is an interesting and practice-relevant variant. %
Again, the results (\Cref{fig:partial-instance-main}) show that for small values of $\alpha$, \ALGOSMART outperforms all classic~algorithms.

\begin{figure}    
    \begin{subfigure}{0.5\textwidth}
        \includegraphics[width=\textwidth]{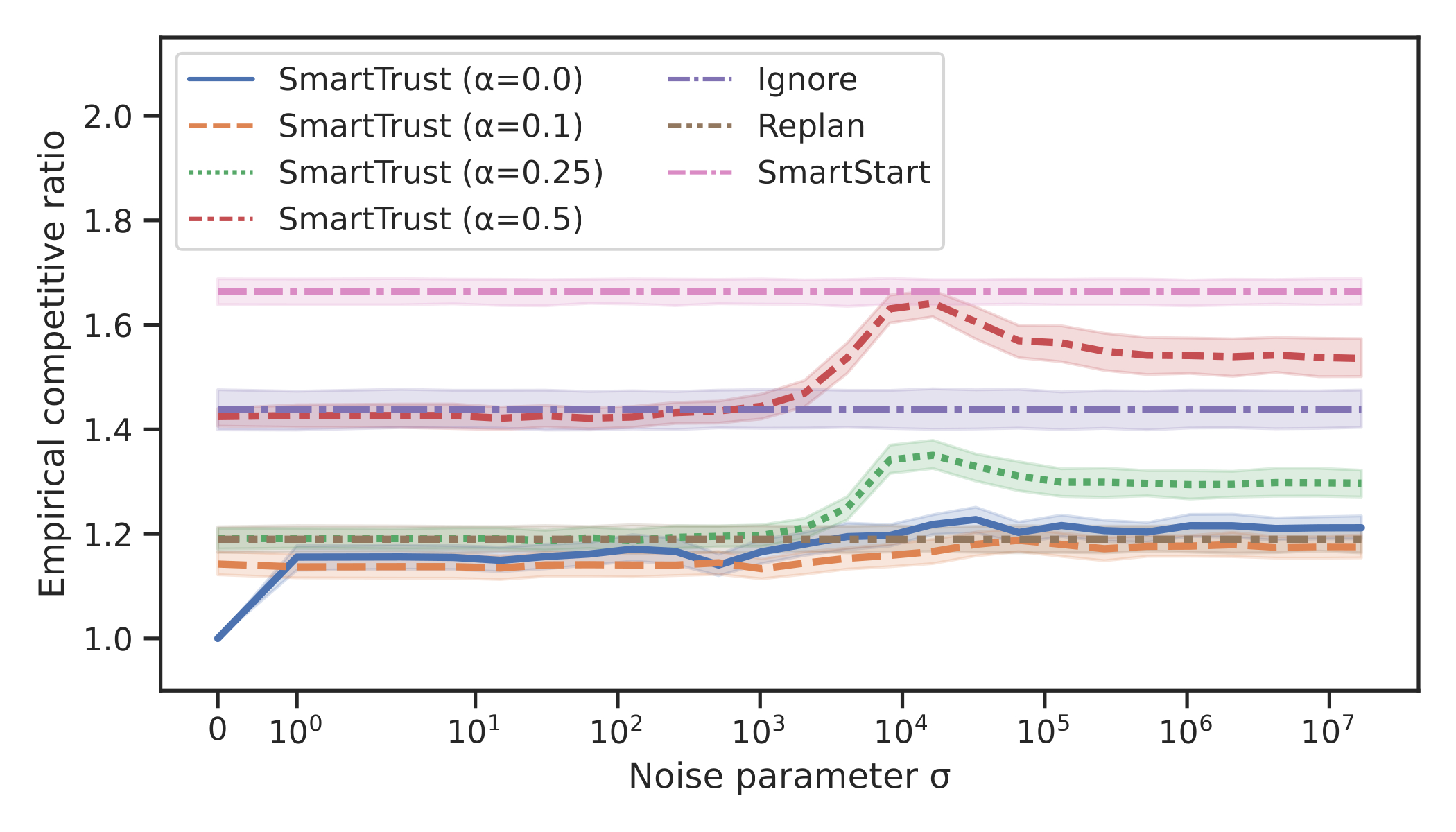}
        \caption{Noise only in request locations}
        \label{fig:noise-in-locations-main}
    \end{subfigure}
    \begin{subfigure}{0.5\textwidth}
        \includegraphics[width=\textwidth]{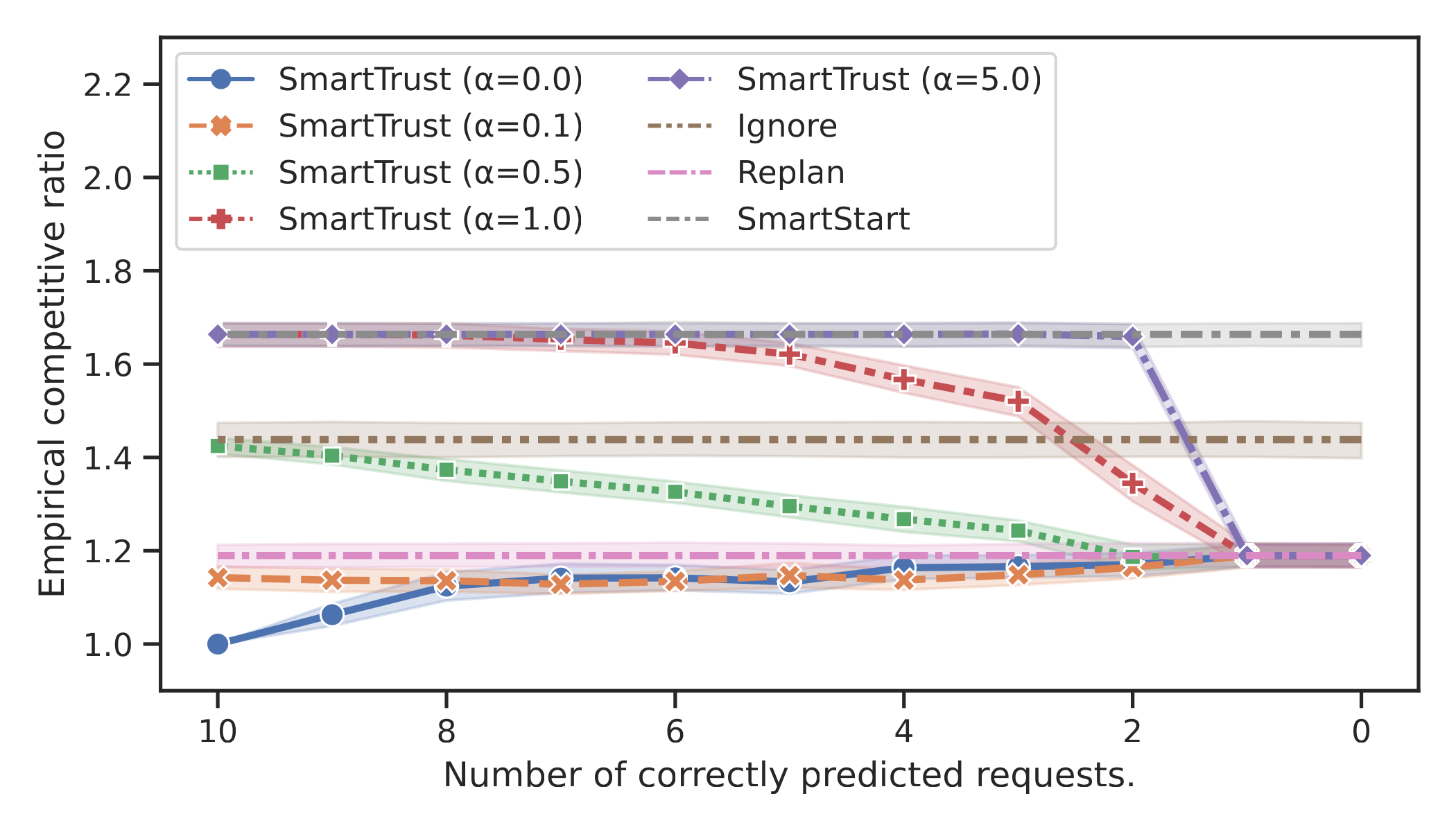}
        \caption{Partial instance predicted correctly}
        \label{fig:partial-instance-main}
    \end{subfigure}
\caption{Experimental results for two different prediction settings  (100 instances with 10 requests each)}
\end{figure}

\section*{Concluding remarks}

The universal cover error can be applied to arbitrary problems with uncertain inputs. %
As it seems to be the first error measure that captures arrival times, it seems very natural to investigate, in particular, other online-time problems such as, e.g., scheduling problems. Further, it would be interesting to identify more compact or smaller predictions. While we predict a full input instance much less information might be sufficient to gain high-quality solutions. This can be only partial information about the input instance or predictions on algorithmic actions, such as an optimal tour instead of request sequences. In the latter case, we can directly apply our framework after approximating the predicted tour by some time-stamped discrete points and using those as input prediction.

\bibliographystyle{plain}
\bibliography{references}

\newpage

\appendix

\section{A lower bound on the consistency-robustness tradeoff}\label{app:lower_bound}

\ThmTradeoffLB*

\begin{proof}
Let~$\epsilon > 0$ be a small constant such that $\epsilon \leq 1 - 2\alpha$. In the following we consider two instances.
The first instance consists of the two requests~$\sigma_1 = (0, 2\alpha + \epsilon)$ and~$\sigma_2 = (1,1)$. 
Since~$\alpha \leq (1 - \epsilon)/2$, 
in the optimal solution the server immediately moves to~1, 
serving $\sigma_2$ at time~1, and is back at the origin at time $2$, serving~$\sigma_1$. 
Suppose that algorithm $\cA$ has access to a perfect prediction. Thus, it has to finish the instance within 
time~$2(1 + \alpha)$ due to its consistency. We can make two observations on the behavior of algorithm~$\cA$.
Firstly, $\cA$ must serve~$\sigma_2$ before~$\sigma_1$, as otherwise, its server must be at the origin at 
time $2\alpha + \epsilon$ and can finish the instance at the earliest at time $2(1 + \alpha) + \epsilon$, 
a contradiction. Secondly, at time~1,~$\cA$'s server cannot be strictly to the left of the point~$1 - 2\alpha$, 
otherwise again its consistency would be contradicted.

Now consider a second instance, consisting of the single request $\sigma = (0, 2\alpha + \epsilon)$. Clearly, 
an optimal solution finishes at time~$2\alpha + \epsilon$. Suppose that algorithm~$\cA$ gets the same prediction 
as in the first instance. Since~$\cA$ is deterministic and the two instances are the same until time~1, it will behave the 
same as in the first instance until time~1. By our observations from the first instance we conclude that at 
time~1,~$\cA$'s server is at least at distance~$1 - 2\alpha$ from the origin, and it has not yet served~$\sigma_1$. 
Thus,~$\cA$ can finish the second instance at the earliest at time~$1 + 1 - 2\alpha$, yielding a robustness factor 
of at least~$1/\alpha - 1$ for an arbitrarily small~$\epsilon$.
\end{proof}

\section{Online routing problems with predictions}

\subsection{An improved algorithm for \oltsp with predictions}
\label{app:smarttrust}

We investigate the particular algorithm \smart by Ascheuer et al.~\cite{AscheuerKR00} to be applied in Phase~(i) of \ALGOGEN. We show that the general framework \ALGOGEN can carefully be adjusted to better exploit the properties of \smart and obtain an improved robustness guarantee. 

\begin{algorithm}
    \caption[]{\smart~\cite{AscheuerKR00}}
    \begin{algorithmic}
        \STATE Whenever the server is %
        at the origin at some time~$t$, compute
        an optimal tour~$S$ of length~$\ell(S)$ serving all the released requests currently unserved.
        If $\ell(S) \leq t$, follow~$S$ while ignoring all the requests that are released in the meanwhile. Otherwise, restart the algorithm at time~$\ell(S)$.
    \end{algorithmic}
\end{algorithm}

Ascheuer et al.~\cite{AscheuerKR00} showed that this algorithm has a competitive ratio of $2$ for \oldarp and, thus, \oltsp. \Cref{thm:algogen} directly implies the following result.
  
  \begin{corollary}\label{cor:easy-algogen-smartstart}
      \ALGOGEN using \smart in Phase~(i) has a competitive ratio bounded from above by 
  \[
      \min\left\{(1 + \alpha) \left(1 + \frac{3 \cdot \errhyp_1}{\optC}\right), 3 + \frac{2}{\alpha}\right\}, 
      \]
      for any~$\alpha > 0$. %
  \end{corollary} 
  
When \smart is used in Phase~(i), we can carefully adjust our \ALGOGEN strategy. We exploit the criteria by which the server waits at the origin and expedite, in this case, immediately to Phase~(ii), to avoid unnecessary waiting time. 
We define the following class of algorithms, 
parameterized by $\alpha > 0$, that we name \ALGOSMART. %
  
\begin{algorithm}
    \caption[]{\ALGOSMART}
    \begin{algorithmic}[1]
        \setAlgoLinenoFormat{\roman}
        \STATE Execute \smart with the following stopping criteria. If \smart decides to follow a tour $S$ of length $\ell(S)$ at time $t$ such that~$t + \ell(S) > \alpha \predC$, go to Phase~(ii). If \smart sleeps or idles at time $\alpha \predC$, go to Phase~(iii).
      \STATE Wait until time at least $\alpha \cdot \predC / 2$, then go to Phase~(iii).
      \STATE Follow the \PR strategy until the end. 
    \end{algorithmic}
\end{algorithm}

  It is easy to verify that Lemma~\ref{lem:competitive_general1} holds for \ALGOSMART. %
  Further, in Lemma~\ref{lemma:metric-oltsp-better-robustness} we show that we can improve upon the robustness factor given in Lemma~\ref{lem:competitive_general2}, which then implies \Cref{thm:algosmart} for \oltsp.

  \begin{restatable}{lemma}{lemmaSmartTrust}\label{lemma:metric-oltsp-better-robustness}
  \ALGOSMART has a %
  competitive ratio of at most $2 + 2 / \alpha$, for any~$\alpha > 0$.
  \end{restatable}

  \begin{proof}
    Let~$\algC$ denote the makespan of~\ALGOSMART's tour.
    If the algorithm terminates during Phase~(i), then the competitive ratio is $2$ by Ascheuer et al.~\cite{AscheuerKR00}. 
    Suppose the algorithm enters Phases~(ii) and~(iii). Let $C_{S}$ be the cost of \smart when serving the whole actual online instance: because \smart is 2-competitive, it holds $C_{S}\leq 2\optC$. Since \ALGOSMART reaches Phase~(iii), is must be~$\alpha \predC < C_{S}$, hence it holds $\alpha \predC <  2 \optC$. 

    Let~$t_{(iii)}$ be the time at which Phase~(iii) starts,~$\tabort$ be the time at which Phase~(ii) begins~(i.e., \smart is aborted at time $\tabort$), and $\tlast$ be the last actual release date. We distinguish two cases.
    \begin{description}
    \item[Case $\boldsymbol{r_{last}} < \boldsymbol{t}_{\textbf{(\textit{iii}})}.$] Since \ALGOSMART reaches Phase~(iii), it also entered Phase~(ii) at time~$\tabort$. Let $\Cabort$ be the length of an optimal tour serving all actual requests that are unserved at time $\tabort$. Note that~$\tabort + \Cabort > \alpha \predC$, as otherwise this tour would have been started in Phase~(i), because of the workings of \smart. We further distinguish two subcases:
        \begin{description}     
            \item[ $\tabort \geq \alpha \cdot \predC / 2.$] In this case, \ALGOSMART does not have to wait in Phase~(ii), and thus we have $\tabort = t_{(iii)}$.
            Observe that $\tabort + \Cabort \leq C_{S} \leq 2 \cdot \optC$.
            Since $\tlast < t_{(iii)} = \tabort$, the length of Phase~(iii) is at most $\predC + \Cabort$, and we conclude
            \[
                \algC \leq \tabort + \Cabort + \predC \leq 2 \cdot \optC + \predC \leq \left( 2 + \frac{2}{\alpha} \right) \cdot \optC.
            \]

            \item[ $\tabort < \alpha \cdot \predC / 2.$] 
            Let $C_{(iii)}$ be the length of an optimal tour of all unserved actual requests at time $t_{(iii)}$. Since we assume $\tlast < t_{(iii)}$, Phase~(iii) takes at most $C_{(iii)} + \predC$ time. Using the fact that Phase~(iii) starts at time $\alpha \predC / 2$ we obtain
            \begin{align*}
                \algC &\leq \frac{\alpha}{2} \predC + C_{(iii)} + \predC  = C_{(iii)} + \left( 1+\frac{\alpha}{2} \right)\predC \\
                &\leq \optC + \left( 1+\frac{\alpha}{2} \right) \frac{2}{\alpha} \optC  = \left( 2 + \frac{2}{\alpha} \right) \optC.
            \end{align*}
        \end{description}

        \item[Case $\boldsymbol{r_{last}} \boldsymbol{\geq} \boldsymbol{t}_{\textbf{(\textit{iii}})}.$] Once the last request has arrived in Phase~(iii) at time $\tlast$, our tour stays fixed. 
        The cost of the algorithm after $\tlast$ is the cost of following the predicted tour, adapted for incorporating the unexpected, yet unserved, requests. This is bounded above by the cost of returning to the origin, following the predicted tour and finally following the optimal tour. Note, that the cost for returning to the origin is at most $\tlast - t_{(iii)}$. 
        Further, $\tlast\leq \optC$ and Phase~(ii) ensures that $t_{(iii)} \geq \alpha \predC / 2$. Hence, the algorithm's makespan satisfies %
        \begin{align*}
            \algC &\leq \tlast + (\tlast - t_{(iii)}) + \predC + \optC
            \leq \tlast + \left(\tlast - \frac{\alpha}{2} \predC\right) + \predC + \optC \\
            &\leq 2 \tlast + \left(1 - \frac{\alpha}{2} \right) \frac{2}{\alpha}  \optC + \optC 
            \leq \left(3 + \left(1-\frac{\alpha}{2}\right)  \frac{2}{\alpha} \right) \optC = \left( 2 + \frac{2}{\alpha} \right)  \optC. \quad \qedhere 
        \end{align*}
    \end{description}
\end{proof}

  We next show that the bound on the robustness of \ALGOSMART given in Lemma~\ref{lemma:metric-oltsp-better-robustness} is tight. 

  \begin{lemma}\label{lemma:metric-smartstart-lb}
    \ALGOSMART has a robustness factor of at least~$2 + 2 /\alpha$, for any~$\alpha > 0$.
    \end{lemma}
    
    \begin{proof}
    Consider $\mathbb{R}$ as metric space. Let $\bm\hat{\sigma} = (-1/2,1/2)$ be the only predicted request, while the only actual request is~$\sigma = (\alpha / 4 + \epsilon, \alpha/4)$. Clearly,~$\optC = \alpha / 2 + 2\epsilon$ and $\predC = 1$. In Phase~(i), \ALGOSMART executes \smart until time~$\alpha \predC = \alpha$. 
    Since the length of the tour that serves $\sigma$ is equal to $\alpha / 2 + 2\epsilon$, when $\sigma$ is released \smart decides to sleep until time~$\alpha / 2 + 2\epsilon$. When \smart wakes up, the condition of Phase~(i) forbids to execute the tour, and \ALGOSMART immediately goes to Phase~(iii). We conclude that at time~$\alpha / 2 + 2\epsilon$ \ALGOSMART is at the origin  and~$\sigma$ is still unserved. Then, the algorithm needs at least~$1 + \alpha / 2 + 2\epsilon$ time to serve~$\sigma$ and follow the predicted tour. This gives a robustness factor of at least
        \[
            \frac{\alpha/2 + 2\epsilon + 1 + \alpha / 2 + 2\epsilon}{\alpha/2 + 2\epsilon} = \frac{2\alpha + 2 + 8 \epsilon}{\alpha + 4\epsilon} \xrightarrow{\epsilon \to 0} 2 + \frac{2}{\alpha}. \quad \qedhere 
        \]
    \end{proof}

\subsection{Algorithms with polynomial running time}
\label{app:polyalgs}

In this section, we adapt \ALGOGEN to run in polynomial time and prove new performance bounds. 

To this end, we first need to choose a polynomial time algorithm $\mathcal{A}$ for Phase~(i) of \ALGOGEN.
There are classic online algorithms, such as \replan~\cite{AusielloFLST01} and \smart~\cite{AscheuerKR00}, which can be implemented to run in polynomial time by using an approximation algorithm to compute solutions for offline metric TSP instances. Using a $\nu$-approximation algorithm for this task, the mentioned online algorithms achieve a competitive ratio of at most~$\frac{3}{2} + \nu$~\cite{AusielloFLST01} resp.~$\frac{1}{4}(4 \nu + 1 + \sqrt{1 + 8 \nu})$~\cite{AscheuerKR00}. 

Further, we have to adapt $\PR$ to run in polynomial time. Since $\PR$ relies on computing tours through requests with release dates, we also have to approximate such tours in polynomial time, unless $P=NP$. To do so one can use an approximation algorithm for offline metric TSP on the set of remaining requests (both predicted and unexpected) disregarding the release times. We then visit these requests and wait, in case of predicted requests, for release times, if needed.

\begin{lemma}\label{lem:crchristofides}
    Given a polynomial time $\nu$-approximate algorithm for offline metric TSP without release times, we obtain a polynomial time  $(1+\nu)$-approximate algorithm for offline metric TSP with release times by following the returned tour and waiting if we arrive at a request early.
\end{lemma}
\begin{proof}
    The given algorithm returns a tour of length at most $\nu \cdot C^*$ in polynomial time. The total time spent waiting for the visited request to arrive is at most $C^*$, because no new requests arrive after time~$C^*$. Therefore the algorithm takes at most $\nu \cdot C^* + C^* = (1+\nu) \cdot C^*$ in polynomial time.
\end{proof}

The polynomial-time \ALGOGEN uses a slight modification of \PR in Phase~(iii). This is necessary, since only replacing optimal solutions by $(1+\nu)$-approximate solutions using~\Cref{lem:crchristofides} when recomputing tours in \PR can invalidate bounds which we require for proving the error-dependency via previously considered excursions. We therefore only recompute a tour through unserved requests using a $(1+\nu)$-approximation algorithm if this does not worsen our targeted error-dependent bound. On the other hand, this still ensures that we always follow a $(1+\nu)$-approximate tour through the remaining requests, which is important for being robust.

\begin{algorithm}
    \caption{Polynomial-Time \PR}
    \begin{algorithmic}[1]
    \STATE Initially, use~\Cref{lem:crchristofides} to compute and follow a tour on the predicted requests $\predR$.
    \STATE Whenever an unexpected request $(x,r) \in R \setminus \predR$ appears, find the predicted request $(\hx,\hr) \in \predR$ which minimizes $\gammatsp(\{(\hx,\hr)\}, (x,r))$. Modify the current remaining tour to the origin as follows. If~$\hx$ is part of the current remaining tour, add an excursion that starts at $\hx$ at some time $t \geq \hr$, serves~$(x,r)$ and finally returns to $\hx$. Otherwise, add an immediate deviation from $p(r)$ to $(x,r)$ and back to $p(r)$ on the current tour. Let~$T_1$ be the computed tour.
    Additionally, compute a new $(1+\nu)$-approximate tour~$T_2$ from~$p(r)$ to the origin through all unserved requests (including $(x,r)$) using~\Cref{lem:crchristofides}. Follow the shorter tour in~$\{T_1,T_2\}$.
    \STATE If the server receives a signal in the origin, terminate.
\end{algorithmic}
\end{algorithm}

Finding the predicted request~$(\hx,\hr)$ for every unexpected request $(x,r)$ in Step~2 can be done efficiently in time~$\cO(\abs{R \setminus \predR} \cdot \abs{\predR})$.

The remaining part of this section is dedicated to the proof of \Cref{thm:competitive_general_polyn}.

\thmPolyRunningTime*

In the following we prove \Cref{thm:competitive_general_polyn} by separately proving the error-dependent bound in~\Cref{lem:poly-smoothness} and the robustness bound in~\Cref{lem:poly_robustness}. We use $\predC$ to denote the length of an optimal tour on the predicted requests and $\appC$ for the length of a tour which is computed by a (fixed) $(1+\nu)$-approximation on the predicted requests using~\Cref{lem:crchristofides}. Hence~$\appC\leq (1+\nu)\predC$.

\begin{lemma}\label{lem:poly-smoothness}
    Given a $\nu$-approximation algorithm for metric TSP, the competitive ratio of polynomial-time \ALGOGEN using a polynomial time $\rho$-competitive online algorithm in Phase~(i) is, for any~$\alpha > 0$%
    , bounded by 
        $(1 + \nu) (1 + \alpha) \left(1 + \frac{3}{2} \cdot \frac{\errhyp_1}{\optC}\right).$
    \end{lemma}

    \begin{proof}
    We first bound~$\predC$. Fix a min-cost $\infty$-hyperedge cover of~$\predR$ by $R$ and an optimal tour~$\optT$ for~$R$. 
    For every hyperedge $\predR'\cup \{(x,r)\}$ in the cover, we  
    extend~$\optT$ by adding the optimal offline \oltsp tour for~$\predR'$ which starts at~$x$ at the time~$t$ at which~$\optT$ serves~$x$. 
    Note that, since~$r \leq t$, the makespan of this subtour is bounded by the cost of $\predR'\cup \{(x,r)\}$. Since every predicted request is covered by at least one hyperedge, the constructed tour serves~$\predR$ and we conclude that
    \(
        \predC \leq \optC + \Gamma_\infty(\predR, R),
    \)
    whence the length of the approximate tour on all predicted requests is bounded by 
    \(
        (1+\nu) \predC \leq (1+\nu)( \optC + \Gamma_\infty(\predR, R)).
    \)
       
    We now bound the makespan of the tour of the algorithm.
    If the algorithm terminates in Phases~(i) or~(ii), its makespan is at most
    \(
         (1+\nu) \cdot \alpha\predC \leq (1+\nu) \cdot \alpha \cdot (\optC + \Gamma_\infty(\predR, R)) \leq (1+\nu) (1 + \alpha) (\optC + \errhyp_\infty).
    \)
    
    Otherwise, the algorithm reaches Phase~(iii). There it first 
    computes a tour that serves all predicted requests of length~$\appC$.
    Everytime an unexpected request $(x',r')$ appears, the tour selection policy of the polynomial-time \PR algorithm ensures that the makespan of the remaining tour increases at most by the length of the computed excursion to serve $(x',r')$. Let $(\hx,\hr)$ be the predicted request that the algorithm computes and uses for the excursion. We distinguish two cases. 
    
    If $\hx$ is not part of this tour, the algorithm immediately deviates from $p(r')$ to serve $(x',r')$ and then returns to $p(r')$. By the triangle inequality, the length of this excursion is bounded by twice the distance between $p(r')$ and $\hx$, plus the cost for optimally serving~$(x',r')$ from $\hx$ when starting at time~$\hr$. Since $\hx$ is not part of the remaining tour, $(\hx,\hr)$ must have already been served at some time $t$ with~$r' \geq t \geq \hr$. Thus, the algorithm's server is at most $r' - \hr$ units away from~$\hx$ at time $r'$. Therefore, the total time incurred for this excursion is bounded by
    \[
        2 \cdot (r' - \hr) + \gammatsp(\{(x',r')\}, (\hx,\hr)) \leq 3 \cdot \gammatsp(\{(x',r')\}, (\hx,\hr)).
    \]
    Note that the inequality is due to the fact that $(x',r')$ can only be served after its release date. 
    In the other case, the algorithm serves $(x',r')$ via an excursion from $\hx$ at some time $t \geq \hr$.
    This takes at most~$\gammatsp(\{(x',r')\}, (\hx,\hr))$ additional time.
    
    We conclude that 
    Phase~(iii) takes time at most
    \(
        \appC + 3 \cdot \Gamma_1(R, \predR).
    \)
    Adding the time for Phases~(i) and~(ii) and using that~$\nu \geq 1$ gives a makespan of at most
    
    \begin{align*}   
        (1 + \alpha) \appC + 3 \cdot \Gamma_1(R, \predR) 
        &\leq (1+\nu)(1 + \alpha) \predC + 3 \cdot \Gamma_1(R, \predR) \\
        &\leq (1+\nu)(1 + \alpha) \left( \optC + \Gamma_\infty(\predR, R) \right) + 3 \cdot \Gamma_1(R, \predR)  \\
        &\leq (1+\nu)(1 + \alpha) \left( \optC + \frac{3}{2} \cdot \errhyp_1 \right). 
    \end{align*}      
  \end{proof}

\begin{lemma}\label{lem:poly_robustness}
Given a polynomial time $\nu$-approximate algorithm to solve metric TSP, polynomial-time \ALGOGEN has a competitive ratio of at most $\rho + (1 + \nu)(1+ \frac{\rho}{\alpha})$, for any polynomial time $\rho$-competitive algorithm used in Phase~(i).
\end{lemma}

\begin{proof}
If the algorithm terminates during Phase~(i) or~(ii), the competitive ratio is $\rho$. We are guaranteed to finish in one of these two phases if $\rho \optC \leq \alpha\appC$.

Now suppose we do not terminate within Phase~(i) or~(ii). Then, $\appC < \frac{\rho}{\alpha} \optC$. Once the last request has arrived at some time $\tlast \leq \optC$, our tour stays fixed.  We distinguish two cases. %
In the first case, the last request arrives before the end of Phase~(ii). In this case, the cost of our tour comprises of the cost for finishing Phase~(ii), which is at most $\alpha \appC$, and the cost of \PR. The latter is bounded by following the approximate tour, incorporating all unserved requests. This is shorter than following the approximate tour on the unexpected yet unserved requests and all predicted requests (served or not). %
The latter tour has optimal length not larger than $\predC + \optC$. Thus, the length of the approximate tour is upper-bounded by $(1 + \nu) (\predC + \optC).$
Adding the cost of finishing Phase
(ii) yields a bound on the total cost of 
\[ 
    \alpha \appC + (1 + \nu) (\predC + \optC) \leq \alpha \appC + (1 + \nu) (\appC + \optC) \leq \left( \rho + (1 + \nu)\left(1+ \frac{\rho}{\alpha}\right) \right)  \optC .
\]

In the second case, the last request arrives in Phase~(iii). In this case the 
cost after $\tlast$ is at most the cost of starting in the position at time $\tlast$ and serving all yet unserved requests using an approximate tour. The optimal cost is bounded by the cost of going back to the origin, which is at most $\tlast-\alpha \appC$, and following an optimal tour, which takes time at most $\predC+\optC$. Hence the total cost of the approximate tour is bounded by 
\begin{align*}
    \tlast + (1 + \nu) \left(\tlast- \alpha \appC + \predC + \optC \right) 
    &\leq \tlast + (1 + \nu)\left(\tlast- \alpha \appC + \appC + \optC \right) \\ 
    &\leq  (3+2\nu)\optC + (1 + \nu)(1-\alpha)\appC  \\
    &\leq  (3+2\nu)\optC + (1 + \nu) \left(\frac{\rho}{\alpha}-\rho \right) \optC  \\
    &=  (3+2\nu)\optC - (1 + \nu)\rho\optC + (1 + \nu)\frac{\rho}{\alpha}\optC.
\end{align*}
Since $\rho \geq 1$, we conclude that this is at most
\begin{align*}
    (2 + \nu)\optC + (1 + \nu)\frac{\rho}{\alpha}\optC 
    \leq \optC + (1 + \nu) \left(1 + \frac{\rho}{\alpha} \right)\optC 
    \leq \left( \rho + (1 + \nu)\left(1+ \frac{\rho}{\alpha}\right) \right)  \optC.\quad   \qedhere 
\end{align*}
\end{proof}

We finally observe that~\Cref{thm:competitive_general_polyn} gives, using the polynomial-time version of the \smart algorithm in Phase~(i) and Christofides' $\frac{3}{2}$-approximation algorithm~\cite{Chr76,serdyukov78} to solve metric TSP instances, for any $\alpha > 0$, a competitive ratio of at most
\[
    \min \left\{ (1+\alpha)\left( \frac{5}{2} + \frac{15 \cdot \errhyp_1}{4 \cdot \optC} \right), 5.152 + \frac{6.629}{\alpha} \right\}.    
\]
\subsection{Online metric Dial-a-Ride with predictions}
\label{app:oldarp}

In this section we show how the techniques we developed in Section~\ref{sec:metric} can be adapted to solve the more general \oldarp.

The \emph{Online Metric Dial-a-Ride Problem} (\oldarp) is a generalization of %
\oltsp where each
request $(x^s, x^d, r)$  consists of a starting location $x^s$, a destination $x^d$ and a release time $r$. %
 To serve a request $(x^s, x^d, r)$, the server must first visit $x^s$, at some time not earlier than $r$, and then $x^d$. %
 {We assume that}
 the server can carry at most one request at the time {and a move from $x^s$ to $x^d$ cannot be interrupted, i.e., there is no storage possible}. %
 In \oldarp \emph{with predictions} we are given additionally predicted requests $\predR$ specified by $(\hx^s, \hx^d, \hr)$.

We first introduce slight modifications of \ALGOGEN which enable it for \oldarp. 
\begin{enumerate}[(a)]
    \item In Phase~(i), \ALGOGEN executes an algorithm that is~$\rho$-competitive for \oldarp.
    \item While in Phase~(i), \ALGOGEN only picks up a request if it can serve it and return to the origin until time~$\alpha \predC$.
    \item \PR only recomputes tours if the server capacity is currently empty.
\end{enumerate}

In the remaining part of this section we highlight how one can adapt the proofs for \Cref{thm:algogen} and \Cref{thm:algosmart} from \oltsp to \oldarp. 

But first we define the cost of a hyperedge for the cover error for \oldarp. By lifting the cost computed by $\gammatsp$ to subinstances of \oldarp (i.e. for transportation requests), we define the cost of a hyperedge $R' \cup \{(x^s,x^d,r)\}$ as 
\[
    \gammadar(R',(x^s,x^d,r)) = \min \left\{ \gammatsp(R', (x^s,r)), \gammatsp(R', (x^d, r + d(x^s, x^d))) \right\} + D,
\]
where $D = \max_{(x^s,x^d,r) \in R \cap \predR} d(x^s,x^d)$.
Further, for a request $(x^s,x^d,r) \in R \cap \predR$ we set
\(
    \gammadar(\{(x^s,x^d,r)\}, (x^s,x^d,r))  = 0.  
\)
We emphasize that we require the additional quantity $D$ in the hyperedge cost only for covering the unexpected actual requests, $R \setminus \predR$, but not for covering absent predicted request, $\predR \setminus R$.

We first proof the error-dependent bound in \Cref{thm:algogen} for \oldarp.

\begin{lemma}
    The competitive ratio of \ALGOGEN for \oldarp is at most $(1 + \alpha) \left(1 + 3 \cdot \frac{ \errhyp_1}{\optC}\right)$, for any~$\alpha \geq 0$. %
\end{lemma}

\begin{proof}
    We first bound~$\predC$. Fix a min-cost $\infty$-hyperedge cover of~$\predR$ by $R$ and an optimal tour~$\optT$ for~$R$. 
    For every hyperedge $\predR'\cup \{(x^s,x^d,r)\}$ in the cover, we 
    extend~$\optT$ by adding the optimal offline \oldarp tour for~$\predR'$ which is considered in $\gammadar(\predR',(x^s,x^d,r))$.
    Since every predicted request is covered by at least one hyperedge, the constructed tour serves~$\predR$ and we conclude that
    \(
        \predC \leq \optC + \Gamma_\infty(\predR, R).
    \)
    
    We now bound the makespan of the tour of the algorithm.
    If the algorithm terminates in Phases~(i) or~(ii), its makespan is at most
    \(
        \alpha \predC \leq \alpha \cdot (\optC + \Gamma_\infty(\predR, R)) \leq (1 + \alpha) \cdot (\optC + \errhyp_1).
    \)
    
    Otherwise, the algorithm reaches Phase~(iii). There it first 
    computes an optimal tour~$\predT$ serving all predicted requests of length~$\predC$. The makespan only increases when unexpected requests arrive.
    To this end, fix a min-cost $1$-hyperedge cover of~$R$ by~$\predR$ 
    and a hyperedge~$\{(x^s,x^d,r)\} \cup \{(\hx^s,\hx^d,\hr)\}$ of this cover. 
    We upper bound the additional cost due to $(x^s,x^d,r)$ by the cost of an excursion which serves $(x^s,x^d,r)$ from the algorithm's current tour. The algorithm might find a faster tour to serve all remaining requests and henceforth uses that.
    In the following, we assume that the minimum in $\gammadar(\{(x^s,x^d,r)\}, (\hx^s,\hx^d,\hr))$ is attained by $\gammatsp(\{(x^s,x^d,r)\}, (\hx^s,\hr))$, and only note differences for the other case.

    We distinguish two cases depending on the algorithm's remaining tour before request~$(x^s,x^d,r)$ arrived. The first case assumes that~$\hx^s$ (resp. $\hx^d$) is not part of this tour. This implies that~$\hx^s$ (resp.~$\hx^d$) must have already been visited at some time~$t$ with~$r \geq t \geq \hr$ (resp. $r \geq t \geq \hr + d(\hx^s, \hx^d)$). We consider an excursion which starts at the next point in time when the server reaches a point~$\hx$ of a predicted request, serves~$(x^s,x^d,r)$ and returns to~$\hx$. We now bound the additional cost for this excursion. If the server follows another excursion at time~$r$, it must have visited point~$\hx$ before this excursion, so especially before time~$r$, but after time~$\hr$ (resp. $\hr + d(\hx^s, \hx^d)$). If the server serves a correctly predicted request at time~$r$, it will reach~$\hx$ (which is in this case the destination point of the currently served request) latest at time~$r + D$. In both cases, the distance between~$\hx^s$ (resp. $\hx^d$) and~$\hx$ is at most~$D + r - \hr$ (resp. $D +  r - (\hr + d(\hx^s, \hx^d))$). By the triangle inequality we observe that the length of the excursion for~$(x^s,x^d,r)$ is bounded by twice the distance from~$\hx$ to~$\hx^s$ (resp.~$\hx^d$) and the cost for serving~$(x^s,x^d,r)$ from~$\hx^s$ at time~$\hr$ (resp. $\hr + d(\hx^s, \hx^d)$),  that is
    \[
        2 \cdot (D + r - \hr) + \gammatsp(\{(x^s,x^d,r)\}, (\hx^s,\hr)) \leq 3 \cdot \gammadar(\{(x^s,x^d,r)\}, (\hx^s,\hx^d,\hr)).
    \] 
    Note that the inequality is due to the fact that $(x^s,x^d,r)$ can only be served after its release date, i.e.,~$r - \hr \leq \gammatsp(\{(x^s,x^d,r)\}, (\hx^s,\hr)) -D$ (resp. $r - (\hr + d(\hx^s, \hx^d)) \leq \gammatsp(\{(x^s,x^d,r)\}, (\hx^d,\hr+ d(\hx^s, \hx^d)))-D $).

    In the second case, the algorithm's server will visit~$\hx^s$ (resp. $\hx^d$) at some later point in time, especially at least once after time $\hr$ (resp. $\hr + d(\hx^s, \hx^d)$). We thus wait %
    until the algorithm reaches~$\hx^s$ (resp. $\hx^d$) at some time~$t \geq \hr$ (resp. $t \geq \hr + d(\hx^s, \hx^d)$), and then serve~$(x^s,x^d,r)$ using at 
    most~$\gammadar(\{(x^s,x^d,r)\}, (\hx^s,\hx^d,\hr))$ additional time.
 
    Since every actual request is covered by one hyperedge, we conclude that 
    Phase~(iii) takes time at most
    \(
        \predC + 3 \cdot \Gamma_1(R, \predR).
    \)
    Adding the time for Phases~(i) and~(ii) gives a makespan of at most
    \[
        (1 + \alpha) \predC + 3 \cdot \Gamma_1(R, \predR) 
        \leq (1 + \alpha) \left( \optC + \Gamma_\infty(\predR, R) + 3 \cdot \Gamma_1(R, \predR) \right)
        \leq (1 + \alpha) \left( \optC + 3 \cdot \errhyp_1 \right).
    \]
\end{proof}

For the robustness bound of~\Cref{thm:algogen}, note that at time~$\tlast$ when the last request arrives the server might be in the processo of serving a ride, which has remaining cost~$C_{ride}$ at time~$\tlast$. While we have to add this value to the time until the algorithm computes its final tour $T_{final}$ due to Modification~(c), we can also subtract it from the length of $T_{final}$, because this particular ride is already served when computing $T_{final}$. Using the other arguments of the proof of \Cref{lem:competitive_general2} for \oltsp, we conclude the stated robustness bound in \Cref{thm:algogen} for \oldarp.

Similarly, one can lift \ALGOSMART to \oldarp and prove \Cref{thm:algosmart} for \oldarp.

\subsection{Online TSP on the half-line with predictions}\label{app:half_line}

In this section we consider the special case of \oltsp in the very restricted metric space $X=\mathbb{R}_{\geq 0}$, the half-line, and show strengthened results.

Note that a request $(x,r)$ cannot be served before time $\max\{x,r\}$; and once it has been served, at least $x$ more time units are needed for returning to the origin. Hence, for a given set of requests,~$R$, any algorithm has a makespan of at least~$\max_{(x,r)\in R}\{r+x,2x\}$. Further, an optimal offline algorithm can find a tour of this value by immediately going to the point $\max_{(x,r)\in R}\{r+x,2x\} / 2$ and then back to the origin, while serving all other requests on the way. We denote this minimal travel time by~$\optC= \max_{(x,r)\in R}\{r+x,2x\}$.

We propose a much more compact prediction model in which we predict a single value $\predC$ instead of individual requests. This value is a prediction on the makespan of an optimal tour~$C^*$. %

By the argumentation above, we can compute for a given value $\predC$ an optimal tour (assuming that the prediction is correct) by moving at time $0$ to $\predC/2$ and then returning to the origin. However, fully trusting the prediction may lead to solutions of unbounded robustness. Moreover, the lower bound on the tradeoff between consistency and robustness in \Cref{thm:tradeoff-lb} holds for \oltsp with predictions, even on the half-line.

This insight also shows us how to interpret the cover error in this prediction setting. From the point of view of an optimal solution, we can always reduce $R$ to the single request $(\optC/2,0)$. Symmetrically, one can interpret the prediction $\predC$ as input prediction~$\predR = \{(\predC/2,0)\}$. The bipartite graph which we then consider for hyperedge costs has therefore only a single vertex on each side.
Recall that the cost of a hyperedge $R' \cup \{x'\}$ should capture the optimal solution for instance $R'$ with respect to $x'$.
Here, the cost of a hyperedge $\{(x_1,r_1)\} \cup \{(x_2,r_2)\}$ is equal to the additional cost of an optimal solution to serve $(x_1,r_1)$ if it can optimally serve $(x_2,r_2)$ for free. All these observations give the following useful property.

\begin{proposition}\label{prop:errhyp-eta}
    $\errhyp_1 = \abs{\predC - \optC}$.
\end{proposition}
\begin{proof}
\[
    \errhyp_1 = \Gamma_1(R, \predR) + \Gamma_1(\predR, R) 
        = 2 \cdot \max \left\{ 0, \frac{\optC}{2} - \frac{\predC}{2} \right\} + 2 \cdot \max \left\{ 0, \frac{\predC}{2} - \frac{\optC}{2} \right\} = \abs{\predC - \optC}. 
\]
\end{proof}

We propose a learning-augmented algorithm that combines the best possible online algorithm Move-Right-If-Necessary (\mrin)~\cite{BlomKPS01} with predictions in a three-stage framework similar to \Cref{sec:blackbox}. %
but carefully exploits the particular functioning of \mrin. The $\frac{3}{2}$-competitive algorithm \mrin by Blom et al.~\cite{BlomKPS01} works as follows. 

\begin{algorithm}
    \caption[]{\mrin~\cite{BlomKPS01}}
    \begin{algorithmic}
        \STATE At any time $t$, %
        if there is an unserved request to the right of the current position of the server, i.e., $(x,r)$ with $r\leq t$ and $x>p(t)$, the server moves to the right with unit speed and, otherwise, it moves towards the origin with unit speed.  
    \end{algorithmic}
\end{algorithm}

Given an instance of \oltsp with prediction $\predC$, we define \ALGOHL, %
a class of algorithms parameterized by $\alpha \in (0, 0.5]$. %
The strategy consists again of three phases.

\begin{algorithm}
    \caption[]{\ALGOHL}
    \begin{algorithmic}[1]
        \setAlgoLinenoFormat{\roman}
        \STATE Execute \mrin until time $\alpha\predC$. Let $p_{(ii)}=p(\alpha\predC)$ be the position of the server at the end of this phase.
        \STATE Move the server to the point $p_{(iii)} = \frac{1}{2}((1-\alpha)\predC + p_{(ii)})$.
        \STATE Execute \mrin again (starting from $p_{(iii)}$).
    \end{algorithmic}
\end{algorithm}

As the main result of this section, we show the following bound on the competitive ratio of \ALGOHL.

\thmHalfLine*

Before  proving the two bounds separately in Lemmas~\ref{lem:ratio_HL1} and~\ref{lemma:half-line-improved-sense}, we state the following useful observation. Later we show that the bound in Lemma~\ref{lem:ratio_HL1} is tight. 
\begin{observation}
	The point $p_{(iii)}$ (start of Phase~(iii)) is not to the left of the server's position $p_{(ii)}$ at the start of Phase~(ii), that is,~$p_{(iii)} \geq p_{(ii)}$.
\end{observation}
     This is because~$p_{(ii)} \leq \alpha \predC$ and $\alpha\in (0,1/2]$ imply
    \[
       p_{(iii)} = \frac{1}{2} \left( (1 - \alpha) \predC + p_{(ii)} \right) \geq \frac{1}{2} \alpha \predC + \frac{1}{2}p_{(ii)} \geq p_{(ii)}.
    \]

\begin{lemma}\label{lem:ratio_HL1}
     \ALGOHL has a %
     competitive ratio of at most $\frac{3}{2\alpha}$, for any $\alpha\in (0,1/2]$.
\end{lemma}
\begin{proof}
    Let $\algC$ be the makespan of a tour determined by \ALGOHL. %
    If the algorithm terminates in Phase~(i), then~$\algC\leq \frac{3}{2} \cdot \optC$, because \mrin is $3/2$-competitive~\cite{BlomKPS01}. Otherwise, Phase~(i) requires~$\alpha  \predC$ time  and Phase~(ii) requires~$p_{(iii)} - p_{(ii)}$ time. By denoting the time used in Phase~(iii) by~$C_{(iii)}$, it follows
    \begin{equation}
        \algC = \alpha\predC + (p_{(iii)} - p_{(ii)}) + C_{(iii)}. \label{eq:half-line-improved-1}
    \end{equation}
    Recall that in Phases~(i) and~(iii) the algorithm follows the \mrin. %
    Consider the execution of \mrin on the same instance. Due to Phase~(ii), in Phase~(iii) \ALGOHL does not have to serve more requests than the ones served by \mrin after time~$\alpha \predC$. But, since \ALGOHL starts Phase~(iii) at point~$p_{(iii)}$ and \mrin is at point~$p_{(ii)}\leq p_{(iii)}$ at time~$\alpha \predC$, Phase~(iii) may take additional time equal to~$p_{(iii)} - p_{(ii)}$ compared to \mrin to move back if there are no requests to the right of~$p_{(ii)}$. As \mrin takes at most~$\frac{3}{2}\optC$ time for the instance, we conclude that
    \begin{equation*}
        \alpha \predC + C_{(iii)} \leq \frac{3}{2} \optC + (p_{(iii)} - p_{(ii)}).
    \end{equation*}
    By Equation~\eqref{eq:half-line-improved-1} and by the definition of $p_{(iii)}$ we obtain
    \begin{equation}
        \algC \leq 2(p_{(iii)} - p_{(ii)}) + \frac{3}{2}\optC \leq (1 - \alpha) \predC + \frac{3}{2} \optC. \label{eq:half-line-improved-2}
    \end{equation}
    Since we assumed that \ALGOHL does not terminate in Phase~(i), the time required by \mrin for the same instance is at least $\alpha  \predC$. Thus,~$\alpha  \predC \leq \frac{3}{2} \optC$, and together with~\eqref{eq:half-line-improved-2} we obtain the claimed bound:
    \[
        \algC \leq \frac{3(1-\alpha)}{2\alpha} \optC + \frac{3}{2}\optC = \frac{3}{2\alpha} \optC.\qedhere
    \]
\end{proof}

\begin{lemma}\label{lemma:half-line-improved-sense}
\ALGOHL has a competitive ratio of at most~$(1+\alpha) \cdot \left( 1 + \frac{\errhyp_1}{\optC} \right)$, for any $\alpha\in (0,1/2]$.
\end{lemma}

\begin{proof}
We can assume that the instance consists of at least one request, as otherwise, we would 
immediately receive an end-signal in the origin and clearly achieve the stated 
competitive ratio. For the case where \ALGOHL reaches Phase~(iii), let $t_{(iii)}$
denote the time when Phase~(iii) begins, that is, by the definition of $p_{(ii)}$ and $p_{(iii)}$,~$t_{(iii)} = \alpha \predC + (p_{(iii)} - p_{(ii)}) = \frac{1}{2}((1+\alpha)\predC - p_{(ii)})$. We denote again by $\algC$ the cost given by \ALGOHL and we distinguish two cases. In each case we prove that $\algC \leq (1+\alpha) (\optC + \abs{\predC - \optC})$, and then assert the stated competitive ratio using~\Cref{prop:errhyp-eta}.

\begin{description}
\item[Case $\boldsymbol{C^* \leq \predC.}$] We start by showing that~$\algC \leq (1+\alpha) \predC$. If the algorithm does not reach Phase~(ii), clearly~$\algC \leq \alpha \predC$.
Now suppose that the algorithm reaches Phase~(ii) (and therefore Phase~(iii)) and the server is at position~$p_{(iii)}$ at time~$t_{(iii)}$.
Observe that, in this case, every request released at time~$t_{(iii)} + \delta$ cannot be to the right of~$p_{(iii)} - \delta$ for any~$0 \leq \delta \leq p_{(iii)}$. 
Indeed, the existence of such a request~$(x,t_{(iii)} + \delta)$ with~$x > p_{(iii)} - \delta$ would imply
\begin{align*}
    \optC &\geq x + t_{(iii)} + \delta \\
    & > p_{(iii)} - \delta + t_{(iii)} + \delta \\
    &= p_{(iii)} + t_{(iii)} 
    = \frac{1}{2}((1-\alpha)\predC + p_{(ii)}) + \frac{1}{2}((1+\alpha)\predC - p_{(ii)})  = \predC,
\end{align*}
a contradiction.
Therefore, in Phase~(iii) the algorithm first serves all the requests to the right of~$p_{(iii)}$ released before time $t_{(iii)}$~(if there is any) and then the server goes straight back to the origin while serving all remaining requests. Since all requests are in the interval~$[0,\optC/2] \subseteq [0,\predC/2]$, the algorithm needs at most~$\predC$ time for phases~(ii) and~(iii), giving a total time of~$(1 + \alpha)\predC$.

We thus obtain the desired bound on the algorithm's cost: 
\[
\algC \leq (1 + \alpha) \predC = (1 + \alpha)\optC + (1 + \alpha) (\predC - \optC).
\]

\item[Case $\boldsymbol{C^* > \predC.}$] In this case the algorithm must enter Phase~(ii), as otherwise~$\algC \leq \alpha \predC < \alpha \optC \leq \optC$ contradicts that~$\optC$ is the optimal makespan. 
Thus, %
the server reaches position~$p_{(iii)}$ at time~$t_{(iii)}$, at the start of Phase~(iii). 
We claim that \ALGOHL spends at most~$B := \max\{\optC - \predC, \optC/2 - p_{(iii)}\}$ time for moving \emph{rightwards or for waiting at the origin after time~$t_{(iii)}$}. 
This allows us to bound the algorithm's makespan, given by the time $t_{(iii)}$ spent on phases~(i) and~(ii), plus the time $p_{(iii)}$ needed to go back to the origin from the point reached at the start of Phase~(iii), plus twice the time spent going to the right %
(or sitting at the origin waiting) in Phase~(iii), as follows:%
\begin{align*}
    \algC &\leq t_{(iii)} + p_{(iii)} + 2 B = \frac{1}{2}((1+\alpha)\predC-p_{(ii)})+ \frac{1}{2}((1-\alpha)\predC+p_{(ii)}) + 2 B\\
    &= \predC + \max \left\{2 \optC - 2  \predC, \optC - (1 - \alpha)  \predC - p_{(ii)} \right\} \\
    &\leq \max \left\{2  \optC - \predC, \optC + \alpha  \predC \right\} \\
    &\leq (1 + \alpha) \optC + (\optC - \predC), 
\end{align*}
which proves the desired bound. 

We prove the claim by contradiction, for which we 
assume that the algorithm exceeds the bound of~$B$ when serving a request~$(x, r)$. This request must exist, as otherwise the server would not move rightwards or wait at the origin. If~$r \leq t_{(iii)}$, the server directly moves from~$p_{(iii)}$ to~$x$ at the beginning of Phase~(iii). Our assumption implies~$x - p_{(iii)} > B$, and thus
\[
    \optC \geq 2x > 2 B + 2 p_{(iii)} \geq 2 \left( \frac{\optC}{2} - p_{(iii)} \right) + 2 p_{(iii)} = \optC,
\]
a contradiction. If~$r > t_{(iii)}$, denote by~$L$ the total time spent by the server moving leftwards between time~$t_{(iii)}$ and~$r$. Thus,~$r - t_{(iii)} - L$ is equal to the time the server spent moving rightwards or waiting in the origin between time~$t_{(iii)}$ and~$r$. Let~$p(r)$ be the position of the server at time~$r$. Note that~$p(r)\leq x$. Due to our assumption,~$x - p(r) + (r - t_{(iii)} - L) > B$. Since the server moved~$L$ units to the left after time~$t_{(iii)}$, it cannot be to the left of point $p_{(iii)} - L$ at time $r$, that is,~$p(r) \geq p_{(iii)} - L$. We conclude
\begin{align*}
    \optC &\geq x + r \\
    &> B - (r - t_{(iii)} - L) + p(r) + r \\
    &\geq (\optC - \predC) - (r - t_{(iii)} - L) + (p_{(iii)} - L) + r \\
    & = \optC - \predC + t_{(iii)} + p_{(iii)} = \optC,
\end{align*}
again a contradiction. \qedhere
\end{description}
\end{proof}

We next show that the robustness factor $\frac{3}{2\alpha}$ is tight for our algorithm. 
\begin{lemma}
    \ALGOHL has a robustness factor at least $\frac{3}{2\alpha}$, for any $\alpha\in (0,1/2]$.
    \end{lemma}
    
\begin{proof}
    Consider an instance consisting of a single 
    request~$\sigma = (\alpha/3, \alpha/3 + \epsilon)$ for a 
    small~$\epsilon > 0$, and suppose that we give \ALGOHL the 
    prediction~$\predC = 1$. The algorithm serves~$\sigma$ at 
    time~$2/3 \cdot \alpha + \epsilon$ and the server is at position~$\epsilon$ at the start of Phase~(ii), 
    i.e., at time~$\alpha$. Thus, it does not receive 
    an end-signal but the server reaches point~$p_{(iii)} = \frac{1}{2}(1-\alpha + \epsilon)$ at time~$t_{(ii)} = \frac{1}{2}(1+\alpha - \epsilon)$. As there are no further 
    requests, it moves back to the origin. 
    Since an optimal solution serves~$\sigma$ at time~$\alpha/3 + \epsilon$ and 
    is back at the origin at time~$2/3 \cdot \alpha + \epsilon$, we conclude 
    that the prediction is not perfect and
    that the robustness ratio of the algorithm is at least
    \[
        \frac{\frac{1}{2}(1+\alpha - \epsilon) + \frac{1}{2}(1-\alpha + \epsilon)}{\frac{2}{3}\alpha + \epsilon} =
    \frac{1}{\frac{2}{3}\alpha + \epsilon} \xrightarrow{\epsilon \to 0} \frac{3}{2\alpha}. 
    \]
\end{proof}

\section{Online network design problems with predictions}
\label{app:networkdesign}

This section covers our results for the following online-list graph problems: online Steiner tree, online Steiner Forest and (capacitated) facility location. We prove new and improved error-dependent performance bounds with respect to the cover error for algorithms based on the general framework introduced by Azar et al.~\cite{AzarPT22}.

In the following subsections we precisely define the problems mentioned above and specify problem-related details for the definition of the cover error. Before that, we give a short overview over the algorithmic framework of Azar et al.~\cite{AzarPT22}, introduce required notation and terminology which strictly follows~\cite{AzarPT22}, and prove an auxiliary lemma which helps to abstract parts of the connection between their algorithm and the cover error.

The framework of Azar et al.\cite{AzarPT22} combines in an iterative manner the execution of a plain online algorithm $\on$ on the actual request set with partial solutions for the predicted input computed using a $\nu$-approximation algorithm for the price-collection variant of the corresponding offline problem. %
The execution of the framework for the $i$th request is called the $i$th iteration. Whenever in an iteration the total cost paid by the online algorithm so far doubles w.r.t. to the last iteration it doubled, 
an offline solution for some part of $\predR$ is added using the subroutine \subr. Such an iteration is called \emph{major} iteration, and we call the sequence of following iterations until the next major iteration (including this) a \emph{phase}. Let~$m$ be the total number of phases for a fixed input.
We denote by~$Q_j$ the set of requests served by the online algorithm in phase $j$, and for any $Q'_j \subseteq Q_j$, $\on_j(Q'_j)$ is the total cost of the online algorithm incurred to serve the requests in $Q'_j$ in phase $j$. We write $\on_j$ for $\on_j(Q_j)$. Note that the online algorithm is always allowed to use the augmented solutions by the prediction with zero cost. Further, $\hB_i$ %
denotes the total cost of the online algorithm until the latest previous major iteration before iteration $i$,
and $\OPT(R')$ denotes the value of an optimal solution for input~$R'$.

\begin{restatable}{lemma}{azarHelperLemma}\label{lemma:azar-bounds}
Suppose that the following two conditions hold for some function $\mu: \mathbb{N} \to \mathbb{R}^+$ and some $k \geq 1$:
\begin{enumerate}[(a)]
    \item $\OPT(\predR) \leq \OPT + \Gamma_\infty(\predR, R)$
    \item If there is a major iteration~$i$ in which all predicted requests become satisfied, we require $\on_j(H) \leq \cO(\mu(k)) \cdot \gamma(H , \hx)$ 
    for any 
    $H \subseteq Q_j$ such that $\abs{H} \leq k$, $j \in \{m-1,m\}$ and $\hx \in \predR$.
\end{enumerate}
Then, the total cost of the framework in~\cite{AzarPT22} is at most
\(
    \cO(1) \cdot \OPT + \cO(\mu(k)) \cdot \errhyp_k.
\)
\end{restatable}

\begin{proof}
    If there exists a major iteration in which all predicted requests become satisfied, let this be the $i$th iteration. Otherwise, let $i$ be the last iteration.
    Iteration~$i$ satisfies the following bounds~\cite[Lemma~2.2]{AzarPT22}:
    \begin{enumerate}
        \item The cost of the first~$i$ iterations is at 
        most~$\cO(1) \cdot \OPT + \cO(\hB_{i-1})$.
        \item If iteration $i$ is not the last iteration, the total cost of the iterations after iteration $i$ is at most~$\cO(1) \cdot \max \{ \on_{m-1}, \on_{m}\}$.
    \end{enumerate}
    
    We first bound~$\hB_{i-1}$. Let~$i' < i$ be the iteration in which~$\hB_{i-1}$ was set.
    Thus, some predicted requests were not satisfied by \subr in iteration $i'$. By the definition of the algorithm, we conclude that the total cost \subr must pay for satisfying all predicted requests in iteration $i'$, that is $c(\subr(\predR, 0))$, is strictly larger than $3 \nu\hB_{i'} = 3 \nu \hB_{i-1}$.
    By \cite[Lemma 2.1]{AzarPT22}, $\subr(\predR, 0)$ approximates the least 
    expensive solution that satisfies all predicted requests within a factor of $3\nu$. 
    This implies that the optimal solution for 
    $\predR$, which is such a solution, has cost of at 
    least $\hB_{i-1}$, i.e. $\OPT(\predR) \geq \hB_{i-1}$.
    Condition~(a) therefore implies that the cost of the first~$i$ iterations is at most
    \[
    \cO(1) \cdot \OPT + \cO(\hB_{i-1}) \leq \cO(1) \cdot \OPT + \cO(1) \cdot \OPT(\predR) \leq \cO(1) \cdot \OPT + \cO(1) \cdot \Gamma_\infty(\predR, R).
    \]
    
    Second, we bound~$ \max\{ \on_{m-1}, \on_{m} \}$ and assume that all predicted requests $\predR$ are satisfied in iteration~$i$. 
    Let~$j \in \{m-1,m\}$ and 
    let $Q_j \subseteq R$ be the subset
    of requests considered by the online algorithm in phase~$j$. 
    Fix a min-cost $k$-hyperedge cover of $R$ by $\predR$, a hyperedge~$R' \cup \{\hx\}$ of the cover and let~$H_j = Q_j \cap R'$ be the subset of 
    requests of the hyperedge which $\on_j$ serves. Recall that~$|R'| \leq k$. 
    Condition~(b) implies
    \[
        \on_j(H_j) \leq \cO(\mu(k)) \cdot \gamma(H_j, \hx),
    \]  
    By the choice of the hyperedge cover, every request in~$Q_j$ is in at least one hyperedge. 
    Summing over all hyperedges of the cover gives
    \[
        \on_j(Q_j) \leq \cO(\mu(k)) \cdot \Gamma_k(R, \predR),
    \] 
    and 
    \[
        \max \{ \on_{m-1}, \on_{m} \} \leq \cO(\mu(k)) \cdot \Gamma_k(R, \predR).
    \]
    
    We conclude the statement by adding up both bounds, that is
    \[
         \cO(1) \cdot \OPT + \cO(\Gamma_\infty(\predR, R)) +\cO(\mu(k)) \cdot \Gamma_k(R, \predR) 
        \leq \cO(1) \cdot \OPT + \cO(\mu(k)) \cdot \errhyp_k. \qedhere
    \]
    \end{proof}

\subsection{Online Steiner tree}

In the online (rooted) Steiner tree problem, a sequence of vertices $R \subseteq V$ (called terminals) of a graph~$G = (V,E)$ with a distinct root vertex $\rho$ is revealed one-by-one. An online algorithm maintains a solution $S$ of selected edges and adds edges to $S$ such that every arriving terminal is connected to $\rho$ via edges in $S$. Every edge $e \in E$ has an associated cost $c_e \in \mathbb{R}^+$, and the objective of the algorithm is to minimize the total cost of selected edges $\sum_{e \in S} c_e$.

We define the cost $\gammast(R', x')$ of a hyperedge~$R' \cup \{x'\}$ as 
the value of the optimal offline Steiner Tree for terminals $R'$ with root~$x'$.

Azar et al.~\cite{AzarPT22} use the Greedy algorithm~\cite{ImaseW91} as online algorithm in their framework, which connects an arriving terminal to the root using the shortest path to the current solution. This algorithm is $\cO(\log \abs{R})$-competitive in the online setting~\cite{ImaseW91}. We denote it by~$\onst$. %

\begin{restatable}{theorem}{hyperedgeSteinerTree}
    The algorithm in~\cite{AzarPT22} for the (undirected) online Steiner tree problem incurs, for every~$k \geq 1$, cost of at most
    \(
        \cO(1) \cdot \OPT + \cO(\log k) \cdot \errhyp_k.
    \)
\end{restatable}

\begin{proof}
    We show the bound on the competitive ratio using \Cref{lemma:azar-bounds}. It remains to 
    prove the two required properties of \Cref{lemma:azar-bounds}.
    
    First, consider a min-cost $\infty$-hyperedge cover of $\predR$ by $R$. 
    For every hyperedge $\predR' \cup \{x\}$ in the cover, we connect~$x \in R$ with
    all predicted requests in~$\predR'$ using the Steiner tree considered in~$\gammast(\predR', x)$. 
    Adding an optimal Steiner Tree for the terminal set~$R$ then also satisfies all requests in~$\predR$ (see \Cref{fig:st-serve-predicted-requests}), and since 
    the augmentation cost are at most the hyperedge costs, we conclude
    \[
        \OPT(\predR) \leq \OPT + \Gamma_\infty(\predR, R).
    \]

\begin{figure}
        \begin{subfigure}[t]{0.63\textwidth}
            \begin{center}           
            \begin{tikzpicture}[scale=0.9]
                \node[origin] (o) at (8,7) {$\rho$};

		\node[corr] (1) at (4,5) {1}; 
		\node[pred] (2) at (3,5) {2}; 
		\node[pred] (3) at (3,7) {3}; 
		\node[pred] (4) at (5,8) {4}; 
		\node[pred] (5) at (6,7.5) {5}; 

		\node[req] (6) at (4,6) {6};
		\node[req] (7) at (8,6) {7};
		\node[req] (9) at (8,8) {9};
		\node[req] (8) at (10,6) {8};
	   
	\draw[reqgreen, line width = 2pt, dashed] (o) -- (9)
	(o) -- (8)
	(o) -- (7) -- (6) -- (1);

	\draw[absent] 
		(9) -- (5) -- (4) 
		(6) -- (3)
		(1) -- (2);
    
            \end{tikzpicture}
            \end{center}
    
            \caption{Augmentation of the optimal Steiner tree for $R$ (dashed) to a Steiner tree which serves $\predR$ using solutions of subinstances (solid) induced by hyperedges of the cover.}
        \end{subfigure}
        \hfill
        \begin{subfigure}[t]{0.35\textwidth}
            \begin{center}     
            \begin{tikzpicture}[scale=0.9]
                \node[pred] (p1) at (0,1 * 0.75) {1}; 
                \node[pred] (2) at (0,2 * 0.75) {2}; 
                \node[pred] (3) at (0,3 * 0.75) {3}; 
                \node[pred] (4) at (0,4 * 0.75) {4}; 
                \node[pred] (5) at (0,5 * 0.75) {5}; 
    
                \node[req] (r1) at (2,1 * 0.75) {1};
                \node[req] (6) at (2,2 * 0.75) {6};
                \node[req] (7) at (2,3 * 0.75) {7};
                \node[req] (8) at (2,4* 0.75) {8};
                \node[req] (9) at (2,5* 0.75) {9};
        
                \node at (2, 0) {$R$};
                \node at (0, 0) {$\predR$};
        
            \begin{pgfonlayer}{background}
                \draw[absent, fill=absentorange!10!,fill opacity=0.5] \convexpath{9,4,5}{9pt};
                \draw[absent, fill=absentorange!10!,fill opacity=0.5] \convexpath{6,3}{9pt};
                \draw[absent, fill=absentorange!10!,fill opacity=0.5] \convexpath{r1,2}{9pt};
            \end{pgfonlayer}
              
            \draw[absent, dotted, line width = 1pt] (r1) -- (p1);
    
            \end{tikzpicture}
            \end{center}
            \caption{Min-cost $k$-hyperedge cover of $\predR$}
        \end{subfigure}
        \caption{}
        \label{fig:st-serve-predicted-requests}
\end{figure}

    Second, 
    assume that there is a major iteration $i$ in which all predicted requests become satisfied. 
    Let~$j \in \{m-1, m\}$, $H_j \subseteq Q_j$ such that~$|H_j| \leq k$ and $\hx \in \predR$.
    In the proof of~\cite[Lemma~3.2]{AzarPT22} it is shown that
    \[
        \onst_j(H_j) \leq \cO(\log k) \cdot \OPT_j(H_j),
    \]  
    where $\OPT_j(H_j)$ is the value of the optimal Steiner Tree for~$H_j$ in which edges 
    that were bought before phase~$j$ have zero cost. Since we assume that $\hx$ has already been served in or before iteration~$i$, a feasible solution for $H_j$ in phase $j$ is given by connecting $H_j$ optimally to $\hx$. Therefore, $\OPT_j(H_j) \leq \gammast(H_j, \hx)$, and thus
    \[    
        \onst_j(H_j) \leq \cO(\log k) \cdot \OPT_j(H_j) \leq \cO(\log k) \cdot \gammast(H_j, \hx). \qedhere
    \]
\end{proof}

\subsection{Online Steiner forest}

In the online Steiner forest problem, a sequence of vertex pairs $R \subseteq V \times V$ (called terminal pairs) of a graph~$G = (V,E)$ is revealed one-by-one. An online algorithm maintains a solution $S$ of selected edges and adds edges to $S$ such that, of every arriving terminal pair $(s,t)$, the vertices $s$ and $t$ are connected only via edges in $S$. Every edge $e \in E$ has an associated cost $c_e \in \mathbb{R}^+$, and the objective of the algorithm is to minimize the total cost of selected edges $\sum_{e \in S} c_e$.

We define the cost $\gammasf(R', (s',t'))$ of a hyperedge~$R' \cup \{(s',t')\}$ as 
the value of the optimal offline Steiner forest for terminal pairs $R'$ where the distance between $s'$ and $t'$ has zero cost.

Azar et al. use a variant of an $\cO(\log \abs{R})$-competitive online algorithm of Berman and Coulston~\cite{BermanC97} in their framework, which we denote by~$\onsf$.

\begin{restatable}{theorem}{hyperedgeSteinerForest}
    The algorithm in~\cite{AzarPT22} for the online Steiner forest problem incurs, for every~$k \geq 1$, cost of at most
    \(
        \cO(1) \cdot \OPT + \cO(k) \cdot \errhyp_k.
    \)
\end{restatable}

\begin{proof}
    We show the bound on the competitive ratio using \Cref{lemma:azar-bounds}. It remains to 
    prove the two required properties of \Cref{lemma:azar-bounds}.
    
    First, consider an optimal solution for $R$ and fix
    a min-cost $\infty$-hyperedge cover of $\predR$ by $R$. 
    For every hyperedge $\predR' \cup \{(s,t)\}$ in the cover, 
    we augment the solution for $R$ with the solution of the subinstance 
    which is considered in $\gammasf(\predR', (s,t))$.
    Note that any terminal pair $(\hs,\htt) \in \predR$ which is connected 
    via connecting to $s$ and $t$ in this solution, remains connected,
    because $s$ and $t$ are connected in the solution of $R$.
    Since every client in $\predR$ is contained in at least one hyperedge,
    the final solution satisfies $\predR$ and we conclude 
    \[
        \OPT(\predR) \leq \OPT + \Gamma_\infty(\predR, R). 
    \]
        
    Second, assume that there is a major iteration $i$ in which all predicted requests become satisfied. Let $j \in \{m-1,m\}$, $H_j \subseteq Q_j$ 
    such that~$|H_j| \leq k$ and $(\hs,\htt) \in \predR$.
    Azar et al.~\cite{AzarPT22} note that for every~$(s,t) \in H_j$, $\onsf$ pays at most twice the distance to any previous request.
    Since we are assuming that~$(\hs,\htt)$ has already been connected before phase~$j$, we conclude that
    $\onsf_j(\{(s,t)\}) \leq 2 \cdot \gammasf(\{(s,t)\}, (\hs,\htt))$. 
    Summing over all requests in $H_j$ gives
    \[
     \onsf_j(H_j) \leq 2k \cdot \max_{(s,t) \in H_j} \gammasf(\{(s,t)\}, (\hs,\htt)) 
    \leq 2k \cdot \gammasf(H_j, (\hs,\htt)) = \cO(k) \cdot \gammasf(H_j, (\hs,\htt)). 
     \]
    \end{proof}

\subsection{Online facility location}

In the online facility location problem, a sequence of vertices $R \subseteq V$ (called clients) of a graph $G = (V,E)$ with edge costs $c_e$ arrive one-by-one. An online algorithm has to connect an arriving client $x$ immediately to the closest open facility $v \in V$ with cost equal to the shortest path between $v$ and $x$. The algorithm is also always allowed to open a (closed) facility $v$ by paying opening cost $f_v$. The objective is to minimize the total connection and opening costs.

For a fixed solution of an instance, we denote for a client $x$ the connected facility by $v_x$.

We define the cost $\gammafl(R', x')$ of a hyperedge~$R' \cup \{x'\}$ as 
the cost for opening facility $x'$ and assigning clients $R'$ to it, that is 
\[
    \gammafl(R', x') = f_{x'} + \sum_{r' \in R'} d(x',r').
\]

We further slightly refine the cost of a $k$-hyperedge cover.
For a $k$-hyperedge cover $\cH'$ of $R_1$ by $R_2$ and for $x_2 \in R_2$, let $\cH'(x_2)$ denote the subset 
of $R_1$ covered by $x_2$ in $\cH'$.
We define the cost of $\cH'$ as follows:
\[
    \sum_{h \in \cH'} \gammafl(h) - \sum_{x_2 \in R_2 : \abs{\cH'(x_2)} = 1} f_{x_2}.
\] %
Intuitively, this modification ensures that the error does not pay opening costs for facilities which cover only a single client in $R_1$, and therefore makes it more sensitive.
We still write $\Gamma_k(R_1, R_2)$ to denote the minimal cost of any such cover. 
Azar et al.~\cite{AzarPT22} use in their framework a $\cO(\log \abs{R})$-competitive online algorithm due to Fotakis~\cite{Fotakis07}, which we denote by $\onfl$.

\begin{restatable}{theorem}{hyperedgeFacilityLocation}
    The algorithm in~\cite{AzarPT22} for the online facility location problem incurs, for every~$k \geq 1$, cost of at most
    \(
        \cO(1) \cdot \OPT + \cO(\log k) \cdot \errhyp_k.
    \)
\end{restatable}

\begin{proof}
    If there exists a major iteration in which all predicted requests become satisfied, let this be the $i$th iteration. Otherwise, let $i$ be the last iteration.
    Iteration~$i$ satisfies the following bounds~\cite[Lemma~2.2]{AzarPT22}:
    \begin{enumerate}
        \item The cost of the first~$i$ iterations is at 
        most~$\cO(1) \cdot \OPT + \cO(\hB_{i-1})$.
        \item If iteration $i$ is not the last iteration, the total cost of the iterations after iteration $i$ is at most~$\cO(1) \cdot \max \{ \on_{m-1}, \on_{m}\}$.
    \end{enumerate}
    
    First, we bound~$\hB_{i-1}$. Let~$i' < i$ be the iteration in which~$\hB_{i-1}$ was set.
    Thus, some predicted requests were not satisfied by \subr in iteration $i'$. By the definition of the algorithm, we conclude that the total cost \subr must pay for satisfying all predicted requests in iteration $i'$, that is $c(\subr(\predR, 0))$, is strictly larger than $3 \nu\hB_{i'} = 3 \nu \hB_{i-1}$.
    By \cite[Lemma 2.1]{AzarPT22}, $\subr(\predR, 0)$ approximates the least 
    expensive solution that satisfies all predicted requests within a factor of $3\nu$. This
    implies that the optimal solution for 
    $\predR$, denoted by $\OPT(\predR)$, which is such a solution, satisfies
    \begin{equation}
        \hB_{i-1} \leq \OPT(\predR).\label{eq:fl-eq1}
    \end{equation}

    We now bound the additional cost in $\OPT(\predR)$ compared to $\OPT(R)$ using the cover error. To this end, 
    consider an optimal solution for instance $R$ and fix a min-cost $\infty$-hyperedge cover~$\cH$
    of $\predR$ by $R$. 
    In the following, we modify the optimal solution for $R$ by both disconnecting some actual requests from their facilities and assigning all requests $\predR$ to existing facilities or new facilities. Note that the latter assignment might not be feasible because clients are not connected to their clostest open facility, but in this case we can clearly compute a feasible assignment without increasing the assignment costs.
    For every hyperedge $\predR' \cup \{x\}$ in $\cH$ we distinguish two cases depending on $\cH(x)$.
    If $\abs{\cH(x)} = 1$, let~$\hx$ be the single request covered by $x$, i.e. $\predR' = \{\hx\}$. In the optimal solution for $R$, we disconnect $x$ from its closest facility $v_x$ and connect
    $\hx$ to $v_x$. Compared to $\OPT(R)$, this increases the cost by at most
    \[
        d(\hx, v_x) - d(x, v_x) \leq d(x, \hx) = \gammafl(\{\hx\}, x) - f_x.
    \]

    If $\abs{\cH(x)} > 1$, we augment the optimal solution for $R$ by the optimal solution considered in 
    the cost function of the hyperedge, i.e. (possibly) open a facility at $x$ and connect all clients in $\predR'$ to $x$. This increases $\OPT(R)$ by at most $\gammafl(\predR', x)$.
    
    Since every request in $\predR$ is covered by at least one
    hyperedge, every client in $\predR$ is assigned to a facility in the final solution and therefore
    \begin{equation}
        \OPT(\predR) \leq \OPT + \Gamma_\infty(\predR, R). \label{eq:fl-eq2}
    \end{equation}
    
    Second, we assume that all predicted requests $\predR$ become satisfied in iteration $i$. We will now bound $\max\{ \onfl_{m-1}, \onfl_{m} \}$. Let $j \in \{m-1,m\}$ and 
    let $Q_j \subseteq R$ be the subset
    of requests considered by the online algorithm in phase~$j$. 
    Fix a min-cost $k$-hyperedge cover $\cH$ of $R$ by $\predR$. 
    Let $F_0$ be the set of open facilities in the beginning of phase~$j$.
    Fix a hyperedge~$R' \cup \{\hx\} \in \cH$ such that $Q_j \cap R' \neq \emptyset$ and let $H_j = Q_j \cap R'$ be the subset of 
    requests of the hyperedge which $\on_j$ serves in phase $j$. Recall that~$|R'| \leq k$. 
    We distinguish two cases. If $\abs{\cH(\hx)} = 1$, let $\{x\} = H_j = R'$ be the request covered by~$\hx \in \predR$ and observe
    that the definition of the amortized costs~\cite[Definition~5.2]{AzarPT22} gives
    \begin{equation*}
        \onfl(x) \leq 2 \cdot d(x, F_0) \leq 2 d(x, \hx) + 2 d(\hx, F_0) = 2 (\gammafl(\{x\}, \hx) - f_{\hx}) + 2 (\hx, F_0). 
    \end{equation*}
    
    If $\abs{\cH(\hx)} > 1$, let $H_j = \{x_{i_1},\ldots,x_{i_{k'}} \}$ 
    be the covered requests indexed
    by their arrival in $R$. Note that $k' \leq k$. 
    For the earliest request in $H_j$ holds $\onfl_j(x_{i_1}) \leq 2 \cdot f_\hx + 2 \cdot d(\hx, x_{i_1})$~\cite[Definition~5.2]{AzarPT22}. 
    For every $2 \leq \ell \leq k'$, %
    let $L_{i_\ell}$ be the set of served requests by $\onfl_j$ until (including) iteration $i_\ell$. Using again~\cite[Definition~5.2]{AzarPT22} gives
    \begin{align*}
    \onfl_j(x_{i_\ell}) &\leq 2 \cdot d(x_{i_\ell}, F_{i_\ell - 1}) \\
    &\leq 2 \cdot d(\hx, F_{i_\ell-1}) + 2 \cdot d(\hx, x_{i_\ell}) \\
    &\leq \frac{2}{\abs{L_{i_\ell-1} \cap H_j}} \left(
    f_\hx + 2 \cdot \sum_{v \in H_j} d(v, \hx) \right) + 2 \cdot d(\hx, x_{i_\ell}) \\
    &\leq \frac{4}{\ell - 1} \cdot \gammafl(H_j,\hx) + 2 \cdot d(\hx, x_{i_\ell}),
    \end{align*}
    where the first inequality is due to the definition of the amortized costs~\cite[Definition~5.2]{AzarPT22}, the second due to the triangle inequality,
    and the third due to an argument by Fotakis~\cite[Corollary~1]{Fotakis07} applied to $\hx$ and $H_j$. 
    Using $k' \leq k$ and summing over all requests in $H_j$ gives
    \begin{align*}
        \onfl_j(H_j) & = \sum_{\ell=1}^{k'} \onfl_j(x_{i_\ell}) \\
        &\leq 2 \cdot f_\hx + \sum_{\ell=2}^{k'} \left( \frac{4}{\ell - 1} \cdot \gammafl(H_j,\hx) \right) + \sum_{\ell=1}^{k'} 2 \cdot d(\hx, x_{i_\ell}) \\
        &\leq \sum_{\ell=2}^{k'} \left( \frac{4}{\ell - 1} \cdot \gammafl(H_j,\hx) \right) + 2 \cdot \gammafl(H_j,\hx) \\
        &\leq \cO(\log k) \cdot \gammafl(H_j,\hx). 
    \end{align*}
    
    Since every request in $Q_j$ is covered by some hyperedge, summing over all hyperedges in $\cH$ gives
    \begin{equation}
        \onfl_j(Q_j) \leq \cO(\log k) \cdot \Gamma_k(R, \predR) 
        + \sum_{\hx \in \predR} d(\hx, F_0). \label{eq:fl-online-cover}
    \end{equation}
    
    Since $F_0$ contains all facilities opened by a solution of \subr in iteration~$i$,
    $\sum_{\hx \in \predR} d(\hx, F_0)$ is at most the connection cost of this solution.
    Using~\cite[Lemma~2.2]{AzarPT22} concludes that the cost of this solution is bounded by
    \[
    \cO(1) \cdot \OPT + \cO(\hB_{i-1}).
    \]
    Using~\eqref{eq:fl-eq1} and~\eqref{eq:fl-eq2} therefore implies
    \[
        \sum_{\hx \in \predR} d(\hx, F_0) \leq \cO(1) \cdot \OPT + \cO(1) \cdot \Gamma_\infty(\predR, R).
    \]

    Thus, the bound in \eqref{eq:fl-online-cover} is at most
    \begin{equation*}        
        \onfl_j(Q_j) \leq \cO(\log k) \cdot \Gamma_k(R, \predR) 
        + \cO(\OPT) + \cO(1) \cdot \Gamma_\infty(\predR, R) \leq \cO(\OPT) + \cO(\log k) \cdot \errhyp_k. 
    \end{equation*}

    We finally conclude the statement by adding up both bounds, that is
    \[
         \cO(1) \cdot \OPT + \cO(1) \cdot \Gamma_\infty(\predR, R) +\cO(1) \cdot \OPT + \cO(\log k) \cdot \errhyp_k
        = \cO(1) \cdot \OPT + \cO(\log k ) \cdot \errhyp_k. \qedhere
    \]

\end{proof}

\subsection{Soft-capacitated facility location}

In the soft-capacitated facility location problems, every facility on a vertex~$v$ has a given capacity $\beta_v$ on the number of clients that can be connected to it. Additionally, a solution is allowed to open multiple facilities, each with cost $f_v$, at a vertex~$v$.

There exists a folklore reduction from the soft-capacitated variant to the standard facility location problem. For details we refer to~\cite{AzarPT22}.

This reduction also implies a reduction between costs of a hyperedge in the capacitated and non-capacitated case, as they are defined by values of solutions of subinstances.
We therefore conclude that our result for the non-capacitated facility location problem naturally extends to the soft-capacitated problem.

\begin{theorem}
    The algorithm in~\cite{AzarPT22} for the online soft-capacitated facility location problem incurs, for every~$k \geq 1$, cost of at most
    \(
        \cO(1) \cdot \OPT + \cO(\log k) \cdot \errhyp_k.
    \)
\end{theorem}

\section{Experiments}\label{app:experiments}

We performed various empirical experiments on real-world instances that demonstrate the benefits of using our algorithm over classic online algorithms in relevant real-world scenarios.
The source code includes instructions to reproduce them. We ran the experiments on an AMD Ryzen 9 3900X processor under Ubuntu~20.04.

\paragraph{Setup}

As metric space we use the shortest path completion of the road network of Manhattan. This network is composed of 4583 nodes, representing intersections and impasses, and 8130 edges, representing street segments. We construct this network with OSMnx~\cite{boeing17osmnx} from OpenStreetMap data~\cite{osm}. 
We generate real-world \oltsp instances by reconstructing taxi requests in Manhattan. The NYC Taxi \& Limousine Commission offers public datasets of taxi trips online~\footnote{\url{https://www1.nyc.gov/site/tlc/about/tlc-trip-record-data.page}}. Specifically, we use the Yellow Taxi Trip Records from January 2021. %
We extract instances of a fixed length and use the pick-up location as request position. Since the dataset only determines a certain taxi zone as location, we randomly select a node of the city road network that is contained in this area.
For the release date of a request, use the relative pick-up time-stamp and scale it to simulate the unit-speed behavior. We assume an average speed of 100 meters per minute.
To generate synthetic predictions, we use independently sampled Gaussian noise of a given standard deviation~$\sigma$. To compute predicted release dates, we add such noise to the actual release dates, while for predicted locations we randomly search a node that has a distance to the actual requested node of approximately the given noise.
We consider three different prediction settings: (i) noise in release dates and locations, (ii) noise only in locations, and (iii) having a randomly selected fraction of a certain size of the actual instance as prediction.

\paragraph{Algorithms} 
We compare \ALGOSMART\ with the classic online algorithms \replan, \ignore and \smart. 
We use the polynomial implementation of \smart as described in~\cite{AscheuerKR00} (this variant uses a different waiting parameter $\theta = \frac{1 + \sqrt{13}}{2} \approx 2.303$ internally).
\replan greedily recomputes and follows a shortest tour from the current server position to the origin through all unserved requests whenever a new request arrives. \ignore computes and then follows an optimal tour through all unserved requests whenever the server is in the origin, but ignores arriving requests while it is following a tour. 
The implementation of \PR in \ALGOSMART ignores predicted requests which are already known to be absent.
All algorithms compute repeatedly optimal or approximate TSP tours, which is known to be NP-hard~\cite{LawlerLRS1985-TSPbook}. In our implementation, 
Christofides' algorithm~\cite{Chr76,serdyukov78} computes a heuristic for TSP tours in all algorithms, which itself uses Blossom V~\cite{kolmogorov09} to compute a min-cost perfect matching. Christofides' algorithm does not respect release dates. To approximate the optimal solution of a (predicted) \oltsp instance, we compute an approximate tour through all requests of the instance using Christofides' algorithm and then follow this tour greedily, i.e.\ wait at every point until the latest request at this point appeared. %

\paragraph{Results}
For every prediction setting, we simulate all algorithms on every instance for every prediction parameter. We compute in every run its empirical competitive ratio, i.e.\ the average ratio between the algorithm's makespan %
and the approximated value of the optimal makespan. %
Further, we report error bars that denote the 95\% confidence interval 
over all instances.

The results for the prediction setting (i) with noise in the release dates and in the locations are visualized in \Cref{fig:res1}. Although \replan performs really well on these real-world instances, \ALGOSMART with $\alpha = 0.0$ and~$\alpha=0.1$ still improves upon \replan while noise is small. For large noise, \ALGOSMART with $\alpha = 0.0$ has an unbounded competitive ratio as expected, as it does not ensure robustness by executing a competitive online algorithm in Phase~(i).
For a noise parameter $\sigma \approx 10^5$, \ALGOSMART achieves its worst performance. Larger noise ($\sigma \geq 10^6$) mainly results in late predicted release dates, since the largest path between two nodes in Manhattan is only $\approx 25 \cdot 10^5$ meters long. Therefore, $\predC$ still grows, and at some point \ALGOSMART simply serves the whole instance in the first phase, which lets its performance coincide with the performance of \smart.

\begin{figure}    
    \begin{subfigure}{0.5\textwidth}
        \includegraphics[width=\textwidth]{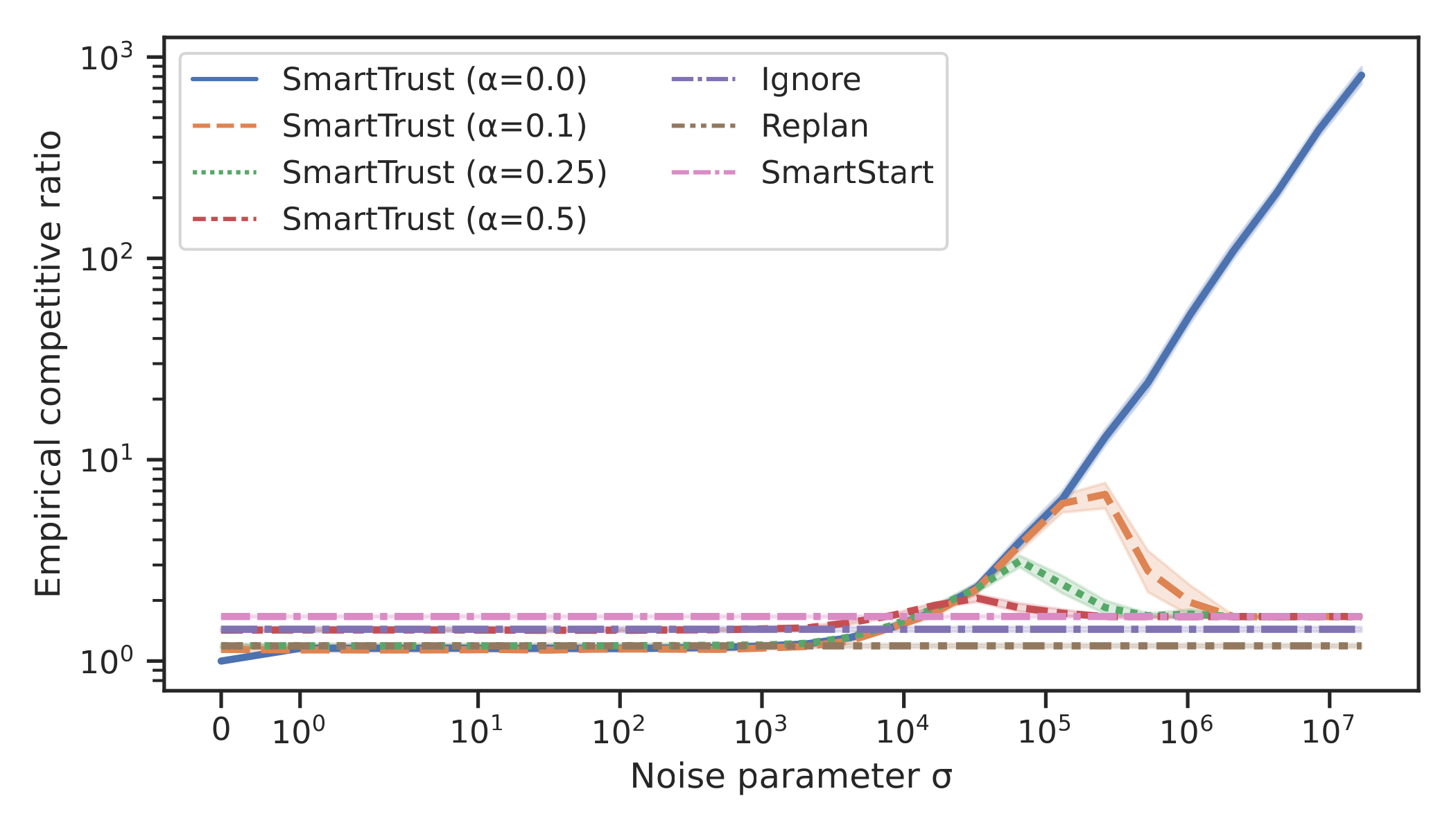}
        \caption{Full picture}
    \end{subfigure}
    \begin{subfigure}{0.5\textwidth}
        \includegraphics[width=\textwidth]{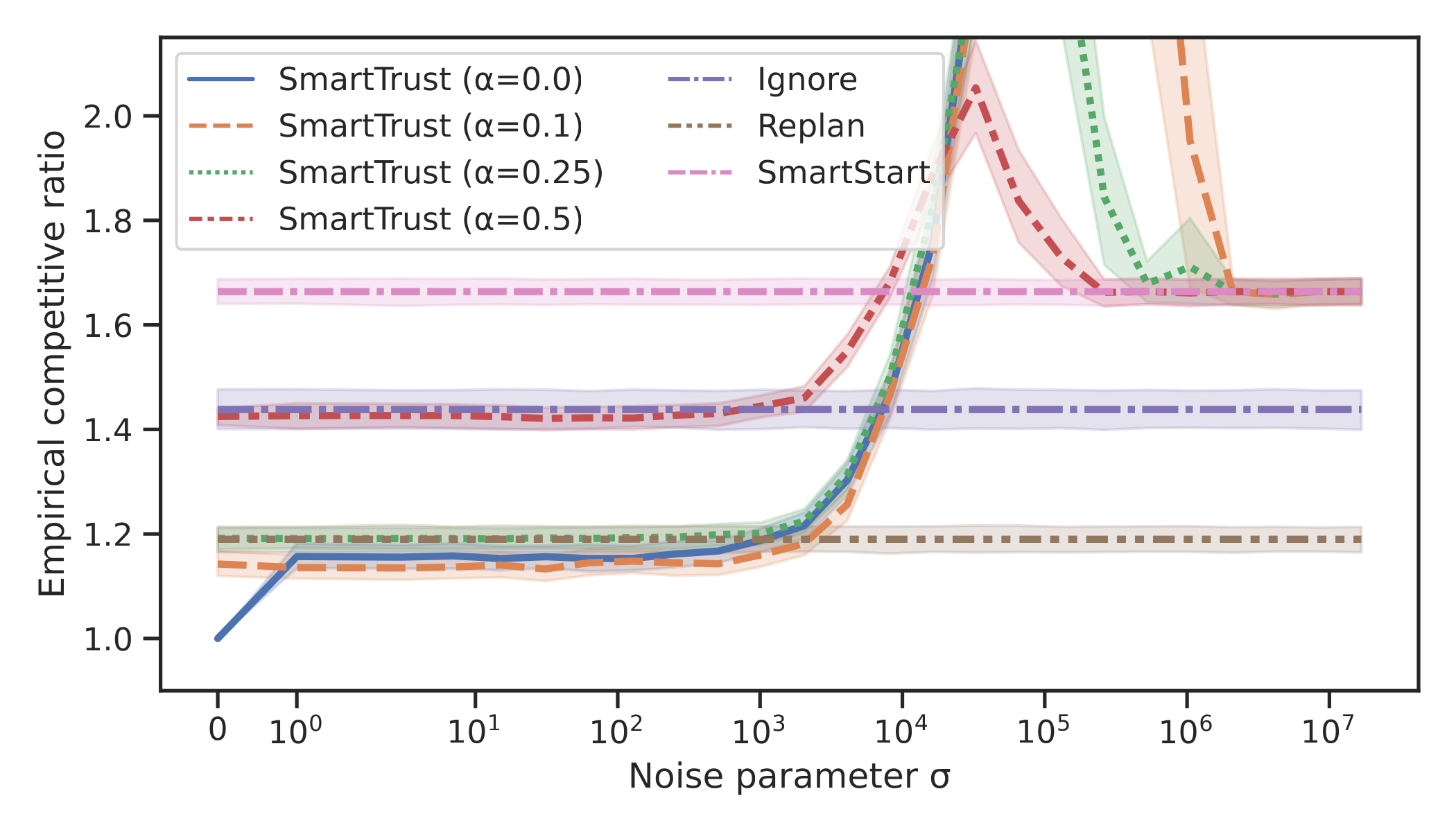}
        \caption{Zoomed-in}
    \end{subfigure}
        \caption{Results when adding noise to the request locations and release dates. (100 instances with 10 requests each)}
        \label{fig:res1}
    \end{figure}

It turns out that mainly erroneous release dates cause this performance spike. Indeed, in the prediction setting (ii) with correct release dates, the results (\Cref{fig:noise-in-locations}) indicate a better performance of \ALGOSMART. Specifically, for $\alpha = 0.1$ it completely dominates \replan over the whole range of reasonable expected prediction quality (by observing that due to the diameter of our network, worse noise parameters would have no significant impact).
Intuitively, in this setting \ALGOSMART with $\alpha > 0$ only slightly pre-moves to predicted locations, but notices absent requests early enough and does not wait until they are predicted to arrive.

The third prediction setting (iii) is a scenario where a part of the actual instance is predicted correctly, and no further erroneous requests are given. One can think of this setting as a taxi company that knows some fixed pick-up requests of the coming day upfront, but still has uncertainty about short-term requests. On the one side, we can observe in \Cref{fig:partial-instance} that for having no predicted requests, the performance of \ALGOSMART equals the performance of \replan, as expected. On the other side, if $\alpha$ is large, we serve a large part of the actual requests in Phase~(i), and achieve a performance close to \smart. For small values of $\alpha$ and some known requests, \ALGOSMART clearly outperforms \replan.

\begin{figure}    
    \begin{subfigure}{0.5\textwidth}
        \includegraphics[width=\textwidth]{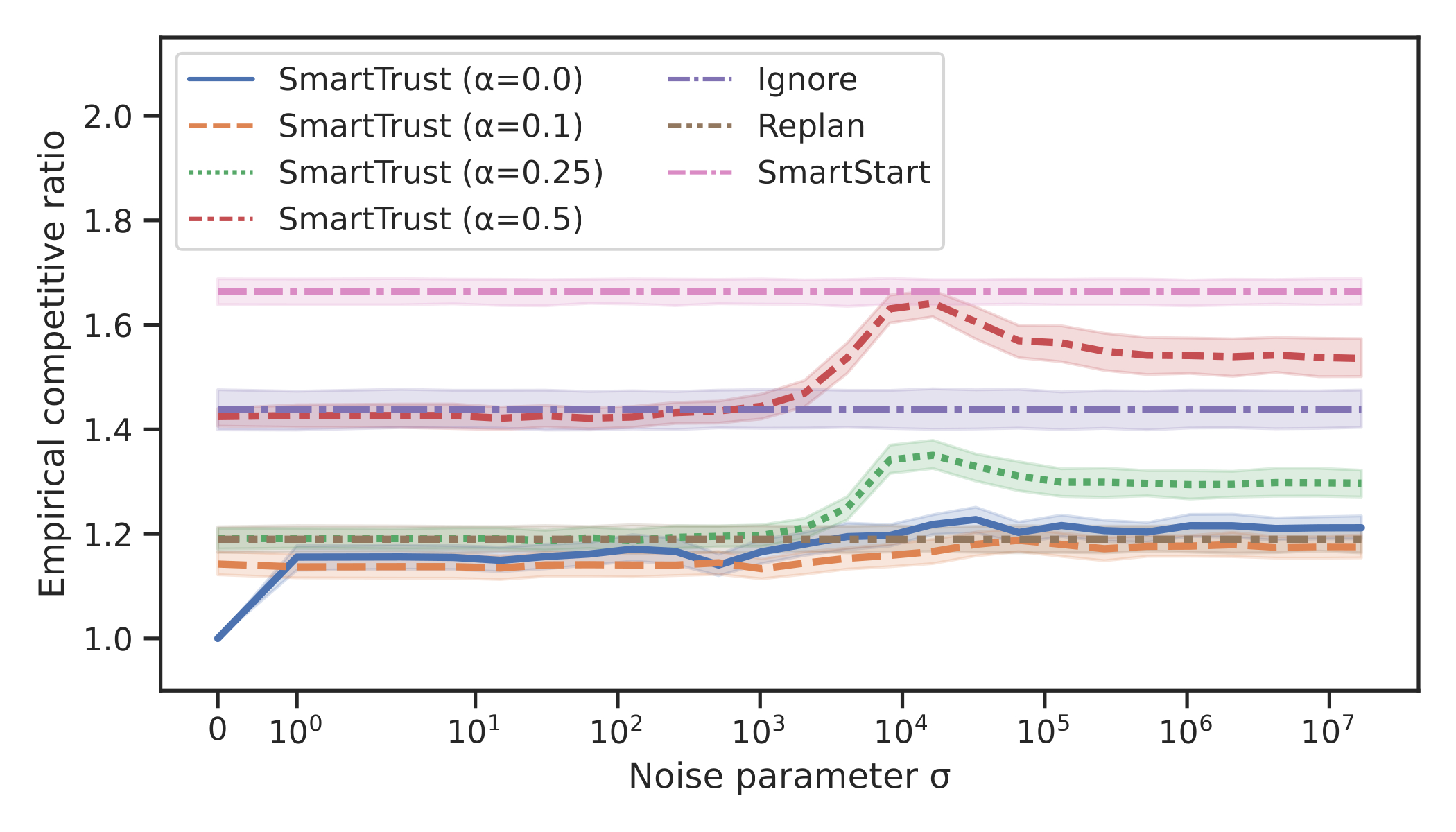}
        \caption{Noise only in request locations}
        \label{fig:noise-in-locations}
    \end{subfigure}
    \begin{subfigure}{0.5\textwidth}
        \includegraphics[width=\textwidth]{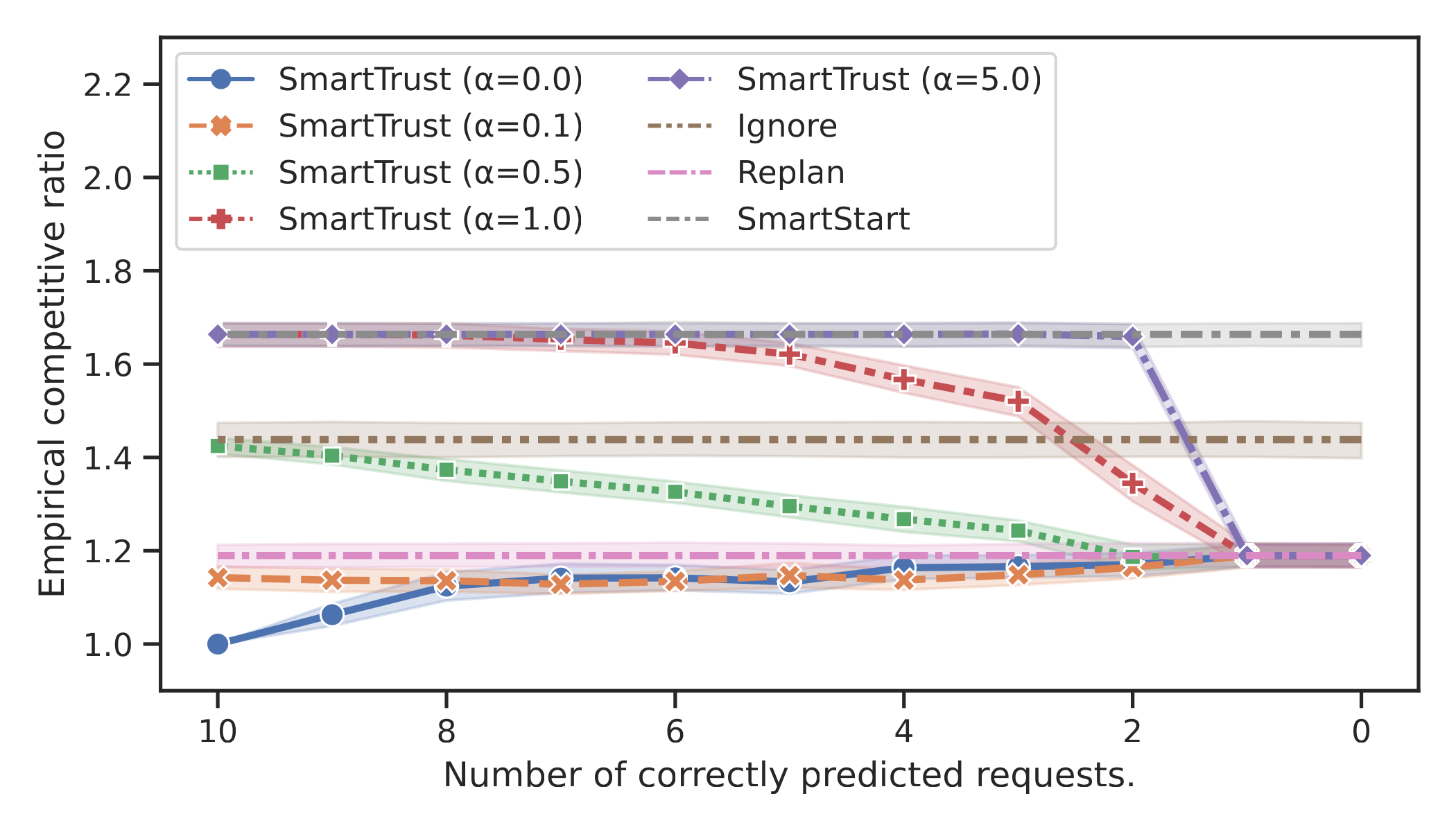}
        \caption{Partial instance predicted correctly}
        \label{fig:partial-instance}
    \end{subfigure}
\caption{Results for prediction settings where parts of the actual instance are predicted correctly. (100 instances with 10 requests each)}
\end{figure}

\end{document}